\documentclass[draftclsnofoot,12pt, onecolumn]{IEEEtran}
\usepackage{cite}
\usepackage{amsmath,amssymb,amsfonts,dsfont}
\usepackage{algorithm}
\usepackage{algorithmic}
\usepackage{mathtools}

\usepackage{graphicx}
\usepackage{subfigure}
\usepackage{textcomp}
\usepackage{xcolor}
\usepackage{booktabs}
\usepackage{breqn}
\usepackage{subeqnarray}
\usepackage{cases}
\usepackage{setspace}
\usepackage{amsthm}
\usepackage{bbm}
\usepackage{bm}
\usepackage{lipsum}
\usepackage{enumitem}

\allowdisplaybreaks[4]

\newtheorem{lem}{Lemma}

\theoremstyle{remark}

\newtheorem{rem}{Remark}

\DeclareMathOperator*{\argmax}{arg\,max}

\def\BibTeX{{\rm B\kern-.05em{\sc i\kern-.025em b}\kern-.08em
    T\kern-.1667em\lower.7ex\hbox{E}\kern-.125emX}}

\begin{document}

\title{Dynamic Task Software Caching-assisted Computation Offloading for Multi-Access Edge Computing\\
}
\author{Zhixiong~Chen,~\IEEEmembership{Student Member,~IEEE},
Wenqiang~Yi,~\IEEEmembership{Member,~IEEE},
Atm~S.~Alam,~\IEEEmembership{Member,~IEEE},
and Arumugam~Nallanathan,~\IEEEmembership{Fellow,~IEEE}
\thanks{Zhixiong Chen, Wenqiang Yi, Atm S.~Alam and Arumugam Nallanathan are with the School of Electronic Engineering and Computer Science, Queen Mary University of London, London, U.K. (emails: \{zhixiong.chen, w.yi, a.alam, a.nallanathan\}@qmul.ac.uk)}
}

\maketitle
\vspace{-1.8cm}
\begin{abstract}
In multi-access edge computing (MEC), most existing task software caching works focus on statically caching data at the network edge, which may hardly preserve high reusability due to the time-varying user requests in practice.
To this end, this work considers dynamic task software caching at the MEC server to assist users' task execution.
Specifically, we formulate a joint task software caching update (TSCU) and computation offloading (COMO) problem to minimize users' energy consumption while guaranteeing delay constraints, where the limited cache size and computation capability of the MEC server, as well as the time-varying task demand of users are investigated.
This problem is proved to be non-deterministic polynomial-time hard, so we transform it into two sub-problems according to their temporal correlations, i.e., the real-time COMO problem and the Markov decision process-based TSCU problem.
We first model the COMO problem as a multi-user game and propose a decentralized algorithm to address its Nash equilibrium solution.
We then propose a double deep Q-network (DDQN)-based method to solve the TSCU policy.
To reduce the computation complexity and convergence time, we provide a new design for the deep neural network (DNN) in DDQN, named state coding and action aggregation (SCAA). In SCAA-DNN, we introduce a dropout mechanism in the input layer to code users' activity states. Additionally, at the output layer, we devise a two-layer architecture to dynamically aggregate caching actions, which is able to solve the huge state-action space problem.
Simulation results show that the proposed solution outperforms existing schemes, saving over 12\% energy, and converges with fewer training episodes.
\end{abstract}

\section{Introduction}
With the development of wireless communications and the proliferation of smart end devices, a large number of computation-intensive applications have emerged to bring powerful functions and ultimate experience to users, such as augmented reality, object recognition, interactive gaming, speech recognition, and natural language processing \cite{9363323}. These applications require massive computational resources and energy. However, the limited computing capability and battery capacity of the mobile devices are generally difficult to meet the computation requirements while executing these applications \cite{mao2017survey}. To cope with it, multi-access edge computing (MEC) has attracted significant attention in industry and academia. MEC deploys cloud-computing capabilities and storage resources within the network edge near to users, such as base stations (BS) and access points (AP) \cite{sabella2016mobile}. It allows mobile users to offload their computation tasks to the network edge with higher computation capability.

\subsection{Related Works}
From the users' perspective, a critical application regarding the MEC is computation offloading (COMO) which is able to save energy and/or speed up the process of computation \cite{mach2017mobile}.
Emerging research towards this direction mainly focus on the joint optimization of the resource allocation and COMO policies.
The authors in \cite{sun2020online} developed an online binary task offloading algorithm to reduce task execution delay in a cellular MEC system.
In \cite{alameddine2019dynamic}, the authors proposed a task offloading and computing resource allocation approach by considering the heterogeneity in the latency requirements of different tasks.
The authors in \cite{yu2020joint} optimized a partial offloading policy in a unmanned aerial vehicle-enabled MEC system to minimizing the task computing delay of clients.
\cite{zhang2020dynamic} studied a joint partial task offloading, computation resource, and radio resource allocation problem to maximize the task computing energy efficiency.
In \cite{zhao2021energy}, the authors investigated an energy consumption minimization problem subject to the latency requirement by optimizing task offloading ratio, transmission power, and subcarrier \& computing resource allocation.

Computing a task requires both the user task data as the input parameters and the corresponding code/task software that processes it. Take face recognition as an example; if a mobile phone needs to identify whether a person is a legitimate user, it takes a photo (input parameters) and uses it as the input data of the face recognition software. After computing, the software output whether the person is a legitimate user, namely computing results.
Existing literature on computation offloading can be classified into two main scenarios: 1) The MEC server has unlimited storage space that can store all task software for users \cite{8467992, 8314696}. In this case, users only need to transmit input parameters to the MEC server for task execution; and 2) The cache size of the MEC server is limited and hence the server fails to cache all task software. Users need to upload both task software and input parameters under this scenario \cite{wen2020joint, yan2021pricing, 9509427, 9076825}. Since the second scenario can be used to characterize most applications in MEC, we consider the second scenario in this work.
The data uploading process and task execution process will generate substantial energy consumption and delay. To improve the computing performance of MEC, caching task computing results at the MEC server has been identified to reduce the frequency of repeated data transmission and task computations \cite{8401954}. It proactively caches some task computing results that may be reused in future task execution \cite{xing2019energy, yang2019joint}.
Although the task computing results caching technique can reduce task execution delay and energy consumption to a certain degree, it is impractical since the task computing results are hardly reusable. In general, computation tasks consist of input parameters and the corresponding task software. The task software is fixed and it can output different computation results under different input parameters. To improve the reusability of cached data, the task software caching technique was proposed to cache the task software at the MEC server to assist the COMO.

Specifically, our previous work \cite{9509427, chen2020dynamic, 9013927} integrated the task program caching mechanism into the COMO technique and designed a model-based task program caching algorithm to minimize the average energy consumption or latency for all time slots.
The authors in \cite{bi2020joint} investigated a single MEC server that assists a mobile user in executing a sequence of computation tasks and used the task program caching technique to reduce the computation delay and energy consumption of the mobile user.
The authors developed an MEC service pricing scheme to coordinate with the service caching decisions and control wireless devices' task offloading behaviours in a cellular network to minimize task execution delay and cost \cite{yan2021pricing}.
The authors in \cite{wen2020joint} provided a joint caching, computation, and communications mechanism to minimize the weighted sum energy consumption subject to the caching and deadline constraints.
In \cite{zhang2018joint}, the authors investigated a joint COMO, content caching, and resource allocation problem in a general MEC network to minimize the total execution latency of computation tasks.
\subsection{Motivation and Contributions}
Existing works on task computing results caching \cite{xing2019energy, yang2019joint} or task software caching-based MEC \cite{9509427, 9076825,  9013927, chen2020dynamic, bi2020joint, yan2021pricing, wen2020joint, zhang2018joint} statically cache data at the network edge, they prefer to cache data that remains unchanged over a relatively long time.
In fact, users' demand for computation tasks dynamically changes over time. The static caching policy cannot preserve the high reusability of the cached data.
Thus, it is important to design learning-based methods to predict the users' task demand and adjust the cache memory dynamically for improving the reusable rate of the cached data.
Moreover, it is noted that most existing works in model-free learning-based content caching design, like \cite{8778670,9110932}, assumed that the task data size is homogeneous, while in practice this assumption does not always hold.
Thus, it is valuable to design a new task software caching update (TSCU) and COMO algorithm which is capable of automatically adapting to the heterogeneous size of task software and dynamically adjust the cache space in real-time according to user requests.

Motivated by this, we consider the dynamic task software caching technique at an MEC  network. Specifically, the task software in the cache memory is updated periodically based on the prediction of users' task computation demand  to assist users' COMO. With the assistance of task software caching, users can accomplish their tasks through either local computing, caching-based COMO, or non-caching-based COMO.
The main contributions of this paper are listed in the following:
\begin{itemize}
  \item We formulate a joint TSCU and COMO problem in a multi-channel wireless environment to minimize the average energy consumption of mobiles users over each time slot while satisfying the task execution delay tolerance. It is intractable to solve its optimal solution due to the lack of user task request information and the complexity of addressing efficient wireless access coordination among multiple users for COMO. With the aid of the maximum cardinality bin packing problem, we theoretically prove that the considered problem is non-deterministic
      polynomial-time hard (NP-Hard).

  \item To tackle this NP-Hard problem, we first decompose it into two distributed sub-problems, i.e., the COMO problem at the user side and the TSCU problem at the MEC server side, and solve them one by one. Since the COMO problem involving a combinatorial optimization over the multi-dimensional discrete space is challenging, we reformulate it as a multi-user COMO game, and theoretically prove the existence of the Nash equilibrium (NE) solution of the COMO game. Based on detailed analysis, We then propose a decentralized algorithm to address its NE solution with a convergence guarantee.

  \item For the second sub-problem, we propose a double deep Q-network (DDQN)-based method to learning the optimal TSCU policy under unknown user task requests information. The massive tasks with heterogeneous data size in the task library result in a high-dimension and complex caching action space which intractable to solve. Moreover, directly using the user request state as the deep neural network (DNN) input may improve the learning complexity. These factors hinder the convergence of the DDQN. To cope with these challenges, we proposed a state coding and action aggregation (SCAA) design for the DNN used in the DDQN. Specifically, we devise a dropout mechanism in the first two layers of the DNN to code users' requests instead of directly using them as input states. A two-layer architecture as the output layer of the DNN dynamically aggregates task software caching action to output the corresponding state-action value. This design effectively reduces the complexity of the DDQN, leading to faster convergence than traditional DDQN algorithms.

   \item We conduct simulations to evaluate the performance of our proposed dynamic TSCU assisted COMO approach. The results show that the proposed approach significantly reduces the users' computation energy consumption. It outperforms the conventional caching update-based COMO approaches. Moreover, the proposed scheme is capable to converge faster than other reinforcement learning-based caching update approaches.
\end{itemize}

\subsection{Organization}
The remaining parts of this paper are organized as follows. In Section \ref{sec:system_model}, we illustrate the system model and formulate the joint TSCU and COMO problem. In Section \ref{sec:Algorithm}, we propose an efficient scheme to solve the original problem. Section \ref{sec:simu} verifies the effectiveness of the proposed scheme by simulations. The conclusion is drawn in Section \ref{sec:conclu}. The code and dataset are available at https://github.com/chfocus/DRL-MEC.

\begin{table}[ht]\tiny\label{tab:notation}
\vspace{-0.6cm}
\small
\caption{Notation Summary}
\vspace{-0.3cm}
\begin{tabular}{|p{1.5cm}|p{6cm}|p{1.5cm}|p{6cm}|}
\hline
Notation & Definition & Notation & Definition \\
\hline
$K$; $F$; $M$  & Number of users; number of tasks; number of subchannels &
$\mathcal{K}$; $\mathcal{F}$; $\mathcal{M}$  & User set; task set; subchannel set\\
\hline
$f_k^{\text{L}}$; $p_k$ & User $k$'s CPU capability; user $k$'s transmit power &
$C$; $f_{\text{C}}$ & MEC server's cache size; MEC server's CPU capability\\
\hline
$B$ & Wireless transmission bandwidth &
$I_f$; $D_f$; $S_f$ & Input parameters' size of task $f$; data size of the task $f$'s software; computation load of task $f$\\
\hline
$\mu_k^{(t)}$ & User $k$'s task request in slot $t$ &
$\alpha_{k,t}$ & User $k$'s COMO decision in slot $t$\\
\hline
$b_f^{(t)}$ & The caching state of the task $f$ in slot $t$ &
$\beta_f^{(t)}$ & Caching update decision of task $f$\\
\hline
$r_{k,t}$ & The uplink transmission rate of user $k$ in slot $t$ &
$\Upsilon_{k,t}$ & Received interference of user $k$ in slot $t$\\
\hline
\end{tabular}
\vspace{-0.8cm}
\end{table}
\section{System Model}\label{sec:system_model}

\subsection{Network Model}
In this paper, we focus on a multi-user MEC network consisting of a BS and $ K $ users as shown in Fig. \ref{fig:sys_model}(a), where the BS is equipped with an MEC server that can access the task library in the cloud centre through an ideal backhaul link. The main notations used throughout this paper are summarized in Table I.
Let $\mathcal{K} = \left\{ {1,2, \cdots ,K} \right\}$ represents the user index set. It is assumed that there are total $ F $ tasks in the task library, whose index set is denoted by $\mathcal{F} = \left\{ {1,2, \cdots ,F} \right\}$.
We consider that the system operates in a sequence of $T$ time slots with an equal length $\tau$. The index set of the time sequence is denoted by $\mathcal{T}=\{1,2,\cdots,T\}$.
The operation mechanism of the system is shown in Fig. \ref{fig:sys_model}(b).
At the beginning of each time slot, each user requests to execute one task in the task library or does not request to execute any task.
Similar to \cite{8491367, 9419788}, we assume that each task must be accomplished before the end of the current slot, either by its local computing or by the MEC server execution. Note that this assumption can be removed by setting delay constraints for each user individually and letting the time slot length be long enough to exceed the maximum delay constraint of users.
Moreover, users' tasks requiring multiple slots to execute are usually inactive in practical system design because this can usually be satisfied by modifying the time slot length.
At the end of this time slot, the MEC server first updates its caching space, and then it caches the selected new task software to assist users' COMO in the next time slot.
After obtaining the task software, the edge server installs the software (e.g., executable .EXE files), and run it based on different input parameters.

\begin{figure}[ht!]
  \centering
  \vspace{-0.6cm}
  \setlength{\abovecaptionskip}{0cm}
  \subfigure[]{\includegraphics[width=0.47\textwidth]{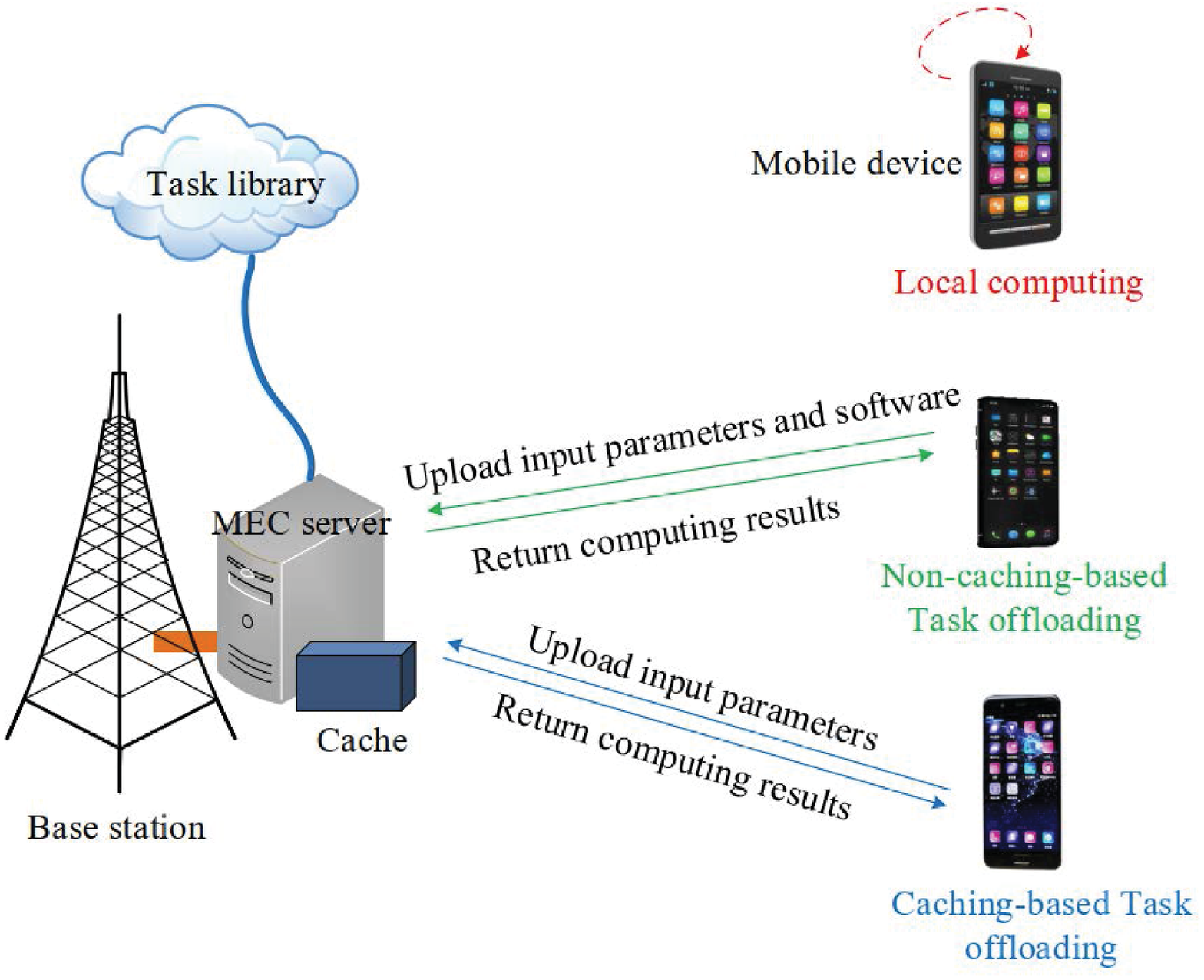}}
  \subfigure[]{\includegraphics[width=0.47\textwidth]{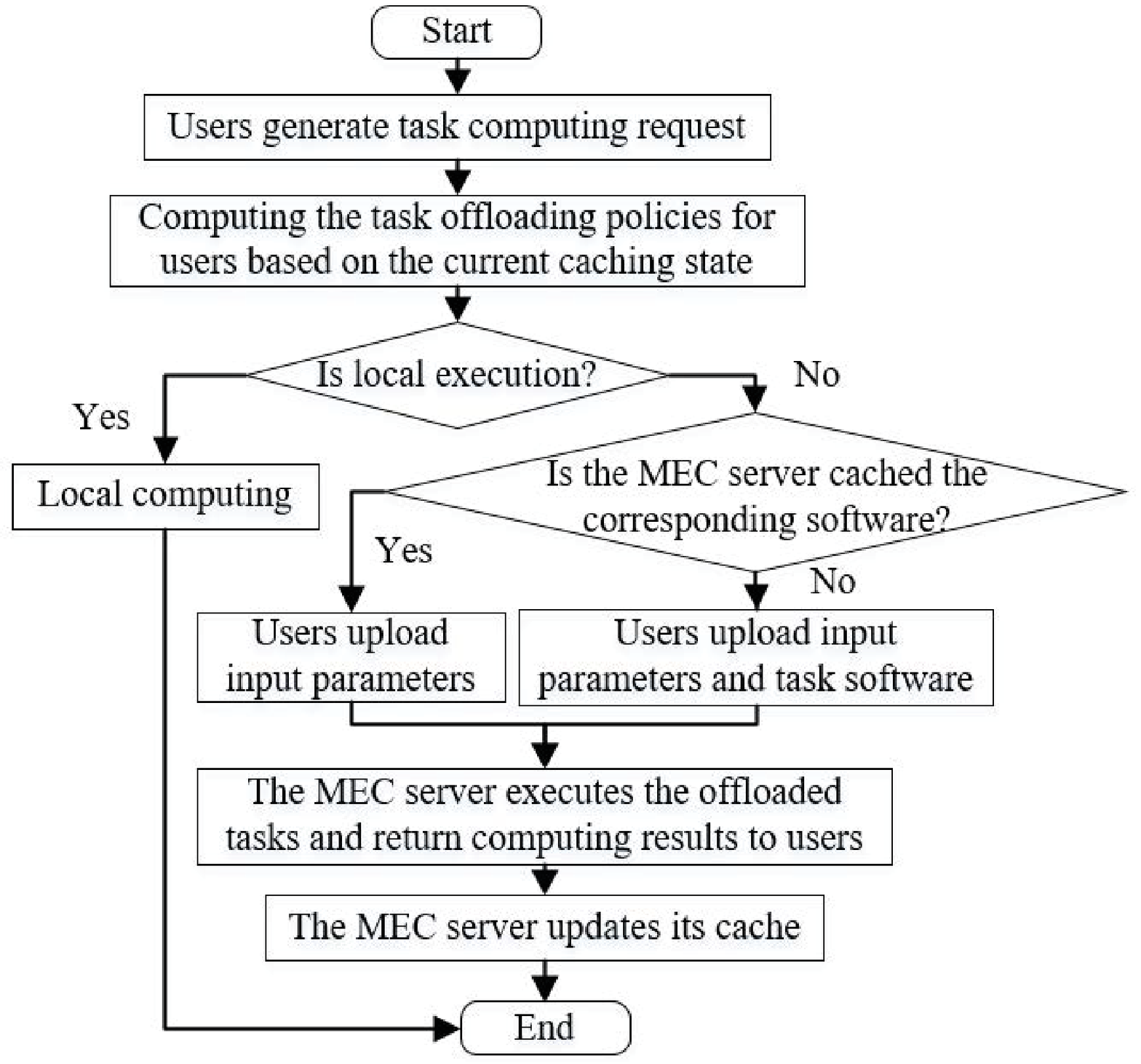}}
  \caption{Illustrating the studied system model: (a) shows the network structure, where one base station is equipped with an MEC server is able to proactively cache selected task software and mobile device has three methods to execute their tasks; and (b) offers the flow chart of the operation mechanism in one time slot.}
  \label{fig:sys_model}
  \vspace{-0.6cm}
\end{figure}

Each task $f \in \mathcal{F}$ can be described by a tuple of three parameters, i.e., $\left\langle {{I_f},{D_f},{S_f}} \right\rangle$, where $I_f$ indicates the size of input parameters of task $f$, $D_f$ is the data volume of the software of task $f$, and $S_f$ denotes the computation load of task $f$, i.e., the necessary central processing unit (CPU) cycles for executing task $f$.
Let $b_f^{(t)} \in \{0, 1\}$ denote the caching state of task $f$ in time slot $t$, where $b_f^{(t)}=1$ represents that the software of task $f$ is cached at the MEC server, $b_f^{(t)}=0$ otherwise.
The caching state in time slot $t$ is characterized by $\bm{b}_t = \{b_1^{(t)},b_2^{(t)},\cdots, b_F^{(t)} \}$. The cache size of the MEC server is denoted by $C$.
Knowing that the cache size is limited, the caching state in any time slot should satisfy
\begin{equation}
\setlength{\abovedisplayskip}{3pt}
\setlength{\belowdisplayskip}{3pt}
\sum\nolimits_{f \in \mathcal{F}} b_f^{(t)} D_f  \le C, \forall t \in \mathcal{T}.
\end{equation}
The TSCU decision profile in time slot $t$ is $\bm{\beta}_t = \{\beta_1^{(t)},\beta_2^{(t)},\cdots, \beta_F^{(t)} \}$.
Let $\beta_f^{(t)} \in \left\{ -1,0,1 \right\}$ indicates the caching update decision for task $f$ in the slot $t$, where $\beta_f^{(t)}=-1$ indicates that the software of task $f$ will be removed at the end of time slot $t$, $\beta_f^{(t)}=0$ denotes that the caching state of $f$ will remain unchanged, and $\beta_f^{(t)}=1$ represents that the software of task $f$ will be added to the cache space in the slot ($t+1$).
Thus, the caching state of task $f$ at the ($t+1$)-th time slot is $b_f^{(t+1)} = b_f^{(t)}+ \beta_f^{(t)}$.
It is noted that $\beta_f^{(t)}$ should satisfy $\beta_f^{(t)} \ge - b_f^{(t)}$ because the MEC server cannot remove uncached task software.

We denote the users' request in time slot $t$ as $\bm{\mu}_t = \{\mu_1^{(t)},\mu_2^{(t)},\cdots,\mu_K^{(t)} \}$. At time slot $t$, let $\mu_k^{(t)} \in \overline{\mathcal{F}}$ ($\overline{\mathcal{F}} = \left\{ 0 \right\} \cup \mathcal{F}$) denote the task request state of user $k$, where $\mu_k^{(t)} = 0$ represents that user $k$ requests nothing, and $\mu_k^{(t)} = f$ ($f \in \mathcal{F}$) indicates that user $k$ requests to execute the task $f$. We assume that $\mu_k^{(t)}$ ($\forall k \in \mathcal{K}$) evolves according to a first-order ($F + 1$)-state Markov chain \cite{8418400} whose transition probability is unknown. That is to say, the users' request in time slot $(t+1)$ is only affected by the users' request in slot $t$  and there are ($F+1$) possible options.

\subsection{Communication Model}
It is assumed that the total available bandwidth in the network is $B$ Hz, which is equally divided into $M$ orthogonal wireless channels. The set of channels is denoted as $\mathcal{M} = \left\{ 1,2, \cdots ,M \right\}$. In each time slot, each user can only use one channel to communicate with the BS. Such a communication method is able to ensure that two users using orthogonal channels do not interfere with each other.
We use $\alpha_{k, t} $ to denote the COMO decision of user $k$ at the $t$-th time slot, where $\alpha_{k, t}=0$ indicates that user $k$ accomplishes its task by its own computing. The $\alpha_{k, t}=m$ ($m \in \mathcal {M}$) denotes that user $k$ selects channel $m$ to offload its task to the MEC server for computing.
We denote the COMO decision of all users in time slot $t$ as $\bm{\alpha}_t = \{\alpha_{1, t},\alpha_{2, t},\cdots, \alpha_{K, t} \}$.
Let $h_k$ and $p_k$ denote the channel gain and transmit power of user $k$, respectively.
In this work, we investigate the task offloading problem under a wireless interference model, in which code division multiple access is deployed to enable multiple users to occupy the same spectrum resource simultaneously for transmitting the information. Thus, the achievable uplink transmission rate of user $k$ in slot $t$ is \cite{rappaport1996wireless, 7307234}
\begin{equation}\label{eq:link_rate}
\setlength{\abovedisplayskip}{3pt}
\setlength{\belowdisplayskip}{3pt}
r_{k,t} = \frac{B}{M} \log \bigg{(} 1 + \frac{p_k h_k}{{\sum\limits_{n \in \mathcal{K} \backslash \{ k\}, \alpha_{n,t} = \alpha_{k,t}} {p_n h_n}  + \sigma^2}} \bigg{)},
\end{equation}
where $\sigma^2$ is the variance of complex white Gaussian channel noise.
In fact, \eqref{eq:link_rate} characterizes the minimal transmit rate of user $k$. The effective interference of user $k$ induced by other users is less than $\sum\nolimits_{n \in \mathcal{K} \backslash \{ k\}, \alpha_{n,t} = \alpha_{k,t}} {p_n h_n}$ and determined by the power control and code design \cite{1194818,chiang2008power}.
Due to the space limits, we investigate the computation offloading problem based the minimal achievable transmit rate in \eqref{eq:link_rate}, and do not consider the power control and code design. Note that, our algorithms designed in the following is able to directly used in the effective channel interference situations.
Moreover, the joint channel code design, power control and computation offloading problem to further improve the offloading performance and manage interference will be a future direction for our work.

From \eqref{eq:link_rate}, users may incur severe interference and low transmission rate when a large number of users offloading theirs tasks through the same channel. As we discuss latter, this would increase the energy consumption for users and forcing part of them to execute tasks by local computing, and thus the number of users in the same channel would be limited.

\subsection{Task Computing}
In our model, we introduce the task software caching mechanism to assist COMO. The MEC server proactively caches the selected task software from the task library and provides computing service for users in the next slot. At the beginning of each time slot, users send their task requests to the MEC server, and then the MEC server returns whether their request tasks are cached. Based on this, when user $k$ needs to execute task $f$, it is able to accomplish $f$ through local computing or caching-based task offloading if $f$ is cached, otherwise through local computing or non-caching-based task offloading. Similar to \cite{8491367, 9419788}, we ignore the information exchange overhead of users acquire whether their task software is cached at the MEC server because it is far small than the input parameters or task software uploading cost.
In the following, we elaborate these three methods:
\begin{enumerate}[label={\arabic*)}]
\item \textbf{Local Computing:} When user $k$ execute its requested task via the local CPU, we denote the computing capability (i.e., CPU cycles per second) of user $k$ ($k \in \mathcal{K}$) as $f_k^{\text{L}}$. Employing the dynamic voltage and frequency scaling technique \cite{mao2017survey}, user $k$ can control the energy consumption for local computing by adjusting the CPU frequency. Considering that user $k$ must finish the local task computing within the current time slot, the CPU frequency of user $k$ satisfies $f_k^{\text{L}} \ge {S_f}/{\tau}$. Based on the realistic measurement result in \cite{miettinen2010energy}, the energy consumption is proportional to the square of the frequency of mobile device. Thus, the energy consumption of user $k$ executes task $f$ by its own device is
    \begin{equation}
    \setlength{\abovedisplayskip}{3pt}
    \setlength{\belowdisplayskip}{3pt}
    E_{k,f}^{\text{L}} = \zeta {(f_k^{\text{L}})^2}{S_f} \ge \zeta \frac{{S_f^3}}{{{\tau ^2}}},
    \end{equation}
    where $\zeta$ is the energy coefficient of mobile devices, determined by the chip architecture. Without loss of the generality, we set the CPU frequency as $f_k^{\text{L}} = S_f/{\tau}$, as this is the most energy-efficient CPU frequency under the deadline constraint. Consequently, The energy consumption of user $k$ executes task $f$ by its own device is $E_{k,f}^{\text{L}} = \zeta \frac{{S_f^3}}{{{\tau ^2}}}$.
\item \textbf{Non-caching-based Task Offloading:} In each time slot $t$, if user $k$ offloads task $f$ to the MEC server for computing, and the MEC server did not cache the corresponding software of task $f$, it needs to upload the input parameters and the corresponding software of task $f$ to the MEC server. In fact, this non-caching-based method is the pure task offloading as illustrated in many existing works, e.g., \cite{sun2020online, alameddine2019dynamic, yu2020joint, zhang2020dynamic, zhao2021energy}.
    Note that, as stated in \cite{9340334},  the MEC server is also able to download the task software from the library each time the request is made by the user $k$, while it only uploads the input parameters. However, the task software acquiring process is time-consuming, especially during peak time. Thus, similar to many existing works, e.g., \cite{wen2020joint, yan2021pricing, 9509427}, we do not allow the edge server to fetch remotely from the library every time the task software is required.
    Let $f_{\text{C}}$ ($f_{\text{C}} \gg f_k^{\text{L}}, \forall k \in \mathcal{K}$) denote the computing capability of the MEC server.
    The task execution delay can be expressed as
    \begin{equation}\label{eq:TKO}
    \setlength{\abovedisplayskip}{3pt}
    \setlength{\belowdisplayskip}{3pt}
    T_{k,f,t}^{\text{O}} = \frac{S_f}{f_{\text{C}}} + \frac{I_f + D_f}{r_{k,t}},
    \end{equation}
    where $r_{k,t}$ follows (\ref{eq:link_rate}). The first part in the right hand side (RHS) of Eq. (\ref{eq:TKO}) is the task execution delay at the MEC server, the second part in the RHS of Eq. (\ref{eq:TKO}) represents the data transmission delay. Considering that the task must be accomplished in the current time slot, the delay should satisfy $T_{k,f,t}^{\text{O}} \le \tau$. The corresponding energy consumption of user $k$ for executing task $f$ is
    \vspace{-0.3cm}
    \begin{equation}\label{eq:non_caching_EC}
    \setlength{\abovedisplayskip}{3pt}
    \setlength{\belowdisplayskip}{3pt}
    E_{k,f,t}^{\text{O}} = p_k \frac{I_f + D_f}{r_{k,t}},
    \end{equation}
    where $r_{k,t}$ is given in (\ref{eq:link_rate}). Note that the energy consumption in (\ref{eq:non_caching_EC}) includes the transmit energy consumption of both input parameters and the corresponding software.
\item \textbf{Caching-based Task Offloading:} When user $k$ offloads the task $f$ to the MEC server for executing in slot $t$, and the MEC server already cached the software of task $f$, it only needs to upload the input parameters and request the MEC server to compute the task $f$ directly and does not need to upload the corresponding software data. Thus, the execution delay can be expressed as
    \vspace{-0.3cm}
    \begin{equation}
    \setlength{\abovedisplayskip}{3pt}
    \setlength{\belowdisplayskip}{3pt}
    T_{k,f,t}^{\text{C}} = \frac{S_f}{f_{\text{C}}} + \frac{I_f}{r_{k,t}}.
    \end{equation}
    Similar to the non-caching-based task offloading method, the execution delay of caching-based task offloading also should satisfy $T_{k,f,t}^{\text{C}} \le \tau$. In addition, the corresponding energy consumption is
    \vspace{-0.3cm}
    \begin{equation}
    \setlength{\abovedisplayskip}{1pt}
    \setlength{\belowdisplayskip}{1pt}
    E_{k,f,t}^{\text{C}} = p_k \frac{I_f}{r_{k,t}},
    \end{equation}
    where $E_{k,f,t}^{\text{C}}$ only includes the transmit energy consumption of the input parameters. Thus, this caching-based task offloading method has lower computational costs (both execution delay and energy consumption) than the non-caching-based task offloading method. Consequently, when user $k$ offloads task $f$ to the MEC server for computing and the software of task $f$ is already cached at the MEC server, there is no doubt that the users will select the caching-based task offloading method for the task execution.
\end{enumerate}
\subsection{Problem Formulation}
In this paper, we aim to minimize the average task execution energy consumption of all users over each time slot under the constraint of task execution delay through jointly optimizing the COMO decision and TSCU policy.
Based on the above models and analysis, we formulate the energy consumption of user $k$ at the $t$-th time slot as
\begin{equation}\label{eq:energy_kt}
\setlength{\abovedisplayskip}{3pt}
\setlength{\belowdisplayskip}{3pt}
E_{k,t} = \sum\limits_{f \in \mathcal{F}} \mathbbm{1}(\mu_k^{(t)} = f) \bigg\{\bigg. \mathbbm{1}(\alpha_{k,t} = 0)E_{k,f}^{\text{L}}
 + \mathbbm{1}(\alpha_{k,t} \in \mathcal{M})\left( {(1 - b_f^{(t)})E_{k,f,t}^{\text{O}} + b_f^{(t)}E_{k,f,t}^{\text{C}}} \right) \bigg\}\bigg.,
\end{equation}
where $\mathbbm{1}(\cdot)$ is an indicator function, which is one if and only if the condition in the parentheses is proper, otherwise it is zero.
Eq. (\ref{eq:energy_kt}) corresponds to three cases: (i) when user $k$ executes the task $f$ through its own device (i.e., $\alpha_{k,t} = 0$), its energy consumption is local computing energy consumption, i.e., $E_{k,t}= E_{k,f}^{\text{L}}$; (ii) when user $k$ executes the task $f$ through COMO and the software has not cached at the MEC server (i.e., $\alpha_{k,t} \in \mathcal{M}$ and $b_f^{(t)}=0$), its energy consumption is $E_{k,t}= E_{k,f,t}^{\text{O}}$ which consists of the transmission energy consumption of input parameters and software; (iii) when user $k$ executes the task $f$ through COMO and the software has already cached at the MEC server (i.e., $\alpha_{k,t} \in \mathcal{M}$ and $b_f^{(t)}=1$), its energy consumption is $E_{k,t}= E_{k,f,t}^{\text{C}}$ which only includes transmission energy consumption of input parameters.
Note that we assume that users will select the caching-based task offloading instead of the non-caching-based task offloading when the corresponding task software has already been cached at the MEC server because the caching-based task offloading method consumes lower energy.
Thus, we can formulate the problem as
\vspace{-0.3cm}
\begin{align}
\mathcal{P}:~~\min_{\bm{\alpha}_t,\bm{\beta}_t}~~&\mathop {\lim}\limits_{T \to \infty } \frac{1}{T}\sum\nolimits_{t = 1}^T \sum\nolimits_{k \in \mathcal{K}} {E_{k,t}}\label{prob:one}\\
\text{s.~t.~~}& \sum\nolimits_{f \in \mathcal{F}} ( b_f^{(t)} +\beta_f^{(t)} ){D_f} \le C, \forall t \in \mathcal{T}, \label{cons:one1}\tag{\theequation a}\\
&\mathbbm{1}(\alpha_{k,t} \in \mathcal{M}) \Big{(} {b_f^{(t)}T_{k,f,t}^{\text{C}}+(1 - b_f^{(t)})T_{k,f,t}^{\text{O}}} \Big{)} \le \tau, \forall k \in \mathcal{K}, \forall f \in \mathcal{F}, \forall t \in \mathcal{T},\label{cons:one2}\tag{\theequation b}\\
&b_f^{(t+1)} = b_f^{(t)}+ \beta_f^{(t)}, \forall f \in \mathcal{F}, \forall t \in \mathcal{T},\label{cons:one3}\tag{\theequation c}\\
&\beta_f^{(t)} \ge - b_f^{(t)}, \forall f \in \mathcal{F}, \forall t \in \mathcal{T}, \label{cons:one4}\tag{\theequation d}\\
&\alpha_{k,t} \in \left\{ {0,1,...,M} \right\}, \forall k \in \mathcal{K},\forall t \in \mathcal{T},\label{cons:one5}\tag{\theequation e}\\
&\beta_f^{(t)} \in \left\{ {-1,0,1} \right\},\forall f \in \mathcal{F}, \forall t \in \mathcal{T}.\label{cons:one6}\tag{\theequation f}
\end{align}
In problem $\mathcal{P}$, (\ref{cons:one1}) implies the cache size constraint of the MEC server.
(\ref{cons:one2}) corresponds to the users' task execution delay restriction.
(\ref{cons:one3}) reveals the TSCU regulations.
(\ref{cons:one4}) indicates that the MEC server cannot remove the uncached task software.
(\ref{cons:one5}) represents the available task computing methods, where $\alpha_{k,t}=0$ indicate that user $k$ executes its task through local computing, and $\alpha_{k,t}=m$ ($ m \in \mathcal{M}$) represents that user $k$ offloads its task (caching-based offloading if $b_f^{(t)}=1$ and non-caching-based offloading if $b_f^{(t)}=0$) through channel $m$.
(\ref{cons:one6}) imposes restrictions on the TSCU decision.
Problem $\mathcal{P}$ is intractable to directly solve since it involves interactive COMO and task software caching across different time slots and lacks user request transition probabilities. We prove it is NP-hard in Lemma \ref{lem:NPhard}.
\begin{lem}\label{lem:NPhard}
Problem $\mathcal{P}$ that involves interactive COMO and TSCU across different time slots is NP-hard.
\end{lem}
\begin{proof}
See Appendix \ref{App:A1}.
\end{proof}

\section{Proposed Computation Offloading and Task Software Caching Update Algorithm}\label{sec:Algorithm}
Due to the intractability of the problem $\mathcal{P}$, one cannot find an effective algorithm to achieve the optimal solution in polynomial time.
In fact, the difficulty of solving problem $\mathcal{P}$ is mainly from the interactive COMO and task software caching across different time slots, as well as the lack of user request transition probabilities.
To cope with these challenges, we decompose the original problem into two subproblems, i.e., the COMO problem and the TSCU problem. First, for any given task software caching state, we reformulate the COMO problem as a multi-user COMO game and then we propose a decentralized algorithm to address its NE solution.
After that, we reformulate the TSCU problem as an Markov decision process (MDP) and use a DDQN to learn the optimal TSCU policy.
\subsection{Multi-user Computation Offloading Algorithm}
Based on the formulation of problem $\mathcal{P}$, the task offloading decision in any time slot $t$ (i.e., $\bm{\alpha}_t$) only affects the energy consumption in $t$, i.e., $E_{k,t}$, and does not related with other slots. In addition, $\bm{\alpha}_t$ does not affect the task software caching decisions in any time slot. Inspired by this, we focus on the COMO problem in a specific time slot $t$ under any given task software caching state $\bm{b}_t$, and design an efficient algorithm to achieve the COMO decision. It is valuable to note that this algorithm can be generalized to solve COMO decisions in any other time slot. We decompose the task offloading problem in slot $t$ from problem $\mathcal{P}$ as:
\begingroup
\setlength{\abovedisplayskip}{0pt}
\setlength{\belowdisplayskip}{2pt}
\begin{align}
\mathcal{P}_1:~~\min_{\bm{\alpha}_t}~~& f_t\left(\bm{\alpha}_t \right) = \sum\nolimits_{k \in \mathcal{K}}  {E_{k,t}} \label{prob:Pro_off}\\
\text{s.~t.~~}&(\text{\ref{cons:one2}}), (\text{\ref{cons:one5}})\notag.
\end{align}
\endgroup
Note that $\bm{\alpha}_t = \{\alpha_{1, t},\alpha_{2, t},\cdots, \alpha_{K, t} \}$, where $\alpha_{k, t}$ ($k \in \mathcal{K}$) has $(M+1)$ value selections. Therefore, the problem $\mathcal{P}_1$ is difficult to solve because it involves a combinatorial optimization over the multi-dimensional discrete space $\{0,1,\cdots, M\}^K$. In the following, we transfer it to a potential game and solve its NE solution.

Let $\alpha_{- k,t} = \{ \alpha_{1,t}, \cdots ,\alpha_{k - 1,t},\alpha_{k + 1,t}, \cdots ,\alpha_{k,t} \}$ denote the task offloading decisions of all other users except from user $k$. The user $k$ is able to choose the optimal computation decision $\alpha_{k,t}^*$ under any given $\alpha_{- k,t}$ in polynomial time with complexity $\mathcal{O}(M+1)$, where $\alpha_{k,t}^* = \mathop {\arg \min }\limits_{\alpha_{k,t}} {f_t}\left( {\alpha_{k,t},\alpha_{- k,t}} \right)$.
Therefore, we transfer the problem $\mathcal{P}_1$ to a multi-user cooperative strategic game $\bm{G} = \left\langle \mathcal{K}, \{ \Lambda_{k,t} \}_{k \in \mathcal{K}}, f_t(\bm{\alpha}_t) \right\rangle$, in which the user set $\mathcal{K}$ is the game player set, $\Lambda_{k,t}$ is the strategy space of user $k$ in time slot $t$ which can be obtained by solving constraint (\text{\ref{cons:one2}}) and (\text{\ref{cons:one5}}), and ${f_t}\left(\bm{\alpha}_t \right)$ is the computing cost of user $k$ (all users have the same computing cost). The objective of game $\bm{G}$ is to achieve a NE solution $\bm{\alpha}_t^* = \left\{ \alpha_{1,t}^*, \cdots ,\alpha_{K,t}^* \right\}$. That is to say, for computation decision $\bm{\alpha}_t^*$ in slot $t$, no user has the ability to further decrease its computing cost through changing its decisions, i.e., ${f_t}( {\alpha_{k,t}^*,\alpha_{- k,t}^*} ) \le {f_t}( \alpha_{k,t},\alpha_{- k,t}^*), \forall k \in \mathcal{K}, \alpha_{k,t} \in \Lambda_{k,t}$.

For any user $k$ ($k \in \mathcal{K}$) in this game $\bm{G}$, it would accomplish its task through task offloading when its local computing cost is larger than task offloading cost, i.e., $E_{k,f}^{\text{L}} \ge (1-b_f^{(t)})E_{k,f,t}^{\text{O}} + b_f^{(t)}E_{k,f,t}^{\text{C}}$. By substituting (3), (5), and (7) into this inequation, we have $\varsigma \frac{S_f^3}{\tau^2} \ge {p_k}\frac{I_f + (1 - b_f^{(t)}){D_f}}{r_{k,t}}$. Let $\Upsilon_{k,t}$ denote the  interference of user $k$,  which satisfies the following inequality:
\begin{equation}\label{ieq:off_not}
\setlength{\abovedisplayskip}{3pt}
\setlength{\belowdisplayskip}{3pt}
\Upsilon_{k,t} = \sum\limits_{n \in \mathcal{K}\backslash \{k\}, \alpha_{n,t} = \alpha_{k,t}} {p_n h_n} \le \frac{p_k h_k}{{2^{\frac{p_k \tau^2 M (I_f + D_f - b_f^{(t)} D_f)} {B \zeta S_f^3}}} - 1} - \sigma^2.
\end{equation}
In other words, for a given task offloading strategy $\bm{\alpha}_t$, the user $k$ is able to decrease the system energy consumption when its received interference satisfies inequation (\ref{ieq:off_not}).
Therefore, if user $k$ received low interference, it decreases its computing cost through task offloading. Otherwise, it accomplishes its task through local computing.
Based on \cite{7307234}, the game $\bm{G}$ is a ordinal potential game by constructing the potential function as follows.
\begin{equation}\label{eq:potential_f}
\setlength{\abovedisplayskip}{3pt}
\setlength{\belowdisplayskip}{3pt}
\phi(\bm{\alpha}_t) = \frac{1}{2}\sum\nolimits_{k = 1}^K  \sum\nolimits_{n \ne k} { p_k h_k p_n h_n \mathbbm{1}(\alpha_{n,t} = \alpha_{k,t})} \mathbbm{1}(\alpha_{k,t} > 0)
+ \sum\limits_{k = 1}^K p_k h_k V_k \mathbbm{1}(\alpha_{k,t} = 0),
\end{equation}
where
\vspace{-0.3cm}
\begin{equation}
\setlength{\abovedisplayskip}{3pt}
\setlength{\belowdisplayskip}{3pt}
V_k = \frac{p_k h_k}{{2^{\frac{p_k \tau^2 M(I_f + D_f - b_f^{(t)} D_f)} {B \zeta S_f^3}}} - 1} - \sigma^2
\end{equation}
is the interference threshold of user $k$ defined according to (11). User $k$ would accomplish its task by task offloading when $\Upsilon_{k,t} \le V_k$, otherwise by local computing.
Note that the change in the potential function (\ref{eq:potential_f}) has the same sign (positive or negative) with the change in the $ f_t\left(\bm{\alpha}_t \right)$.
In Remark \ref{rem:NE_exist}, we prove that the game $\bm{G}$ with the potential function $\phi(\bm{\alpha}_t)$ is a ordinal potential game and it has a NE solution.

\begin{rem}\label{rem:NE_exist}
The COMO game $\bm{G}$ with the potential function $\phi(\bm{\alpha}_t)$ is a ordinal potential game and is able to achieve a NE solution in finite number of iterations.
\end{rem}
\begin{proof}
See Appendix \ref{App:A2}
\end{proof}

Based on Remark \ref{rem:NE_exist}, we develop a potential game-based multi-user COMO algorithm to address a mutually satisfactory offloading decisions (i.e., the NE solution) for all users. The detailed steps of COMO algorithm are summarized in \textbf{Algorithm 1}.
\begin{algorithm}
\algsetup{linenosize=} \small
\caption{Multi-user Computation Offloading}
\begin{spacing}{1.3}
\begin{algorithmic}[1]
\STATE Each user $k \in \mathcal{K}$ initialize its COMO decision $\alpha_{k,t} = 0$
\REPEAT
\FOR{Each user $k \in \mathcal{K}$:}
    \STATE Measure the interference $\Upsilon_{k,t}$ and calculate the transmission rate $r_{k,t}$,
    \STATE Compute the strategy space $\Lambda_{k,t}$ by solving constraint (\text{\ref{cons:one2}}) and (\text{\ref{cons:one5}}),
    \STATE Select the best offloading decision $\alpha_{k,t}^* = \mathop{\arg \min}\limits_{\alpha_{k,t} \in \Lambda_{k,t}} f_t\left( \alpha_{k,t},\alpha_{- k,t}\right)$
    \IF{$\alpha_{k,t}^* \ne {\alpha_{k,t}}$}
        \STATE Send a request message to BS for updating its offloading decision
        \IF{Received the update message}
            \STATE Update its COMO decision, i.e., ${\alpha_{k,t}} = \alpha_{k,t}^*$
        \ENDIF
    \ENDIF
\ENDFOR
\UNTIL Receive an end message
\RETURN $\bm{\alpha}_t$.
\end{algorithmic}
\end{spacing}
\end{algorithm}

Through \textbf{Algorithm 1}, we achieve a NE solution for the COMO problem. Firstly, we initialize the COMO decisions of all users to 0. Next, each user computes its available task offloading decision set $\Lambda_{k,t}$ based on constraints (\text{\ref{cons:one2}}) and (\text{\ref{cons:one5}}), and finds its optimal COMO decision $\alpha_{k,t}^*$. Then, user $k$ sends a update request message to the MEC server if $\alpha_{k,t}^* \ne \alpha_{k,t}$. When the MEC server receives the update request messages from users, it randomly selects one user and then sends the update permission message to this user. The user who receives the update permission message updates its offloading decision, and the users who do not receive the update permission message remain their offloading decisions. Finally, if the MEC server does not receive any update request message from users, it sends the end messages to all users. When users receive the end message, they offload their tasks based on their offloading decisions.
We analyze the convergence behaviour of \textbf{Algorithm 1} in Lemma \ref{lem:gameComplexity}.
\begin{lem}\label{lem:gameComplexity}
Game $\bm{G}$ can achieve a NE solution within $\frac{{\frac{1}{2}{K^2}\Delta_{\max}^2 + K(\Delta_{\max} V_{\max} - \Delta_{\min} V_{\min})}}{{\varepsilon \Delta_{\min}}}$ iterations, where $\varepsilon$ is a positive number.
\end{lem}
\begin{proof}
See Appendix \ref{App:A3}
\end{proof}

\subsection{Deep Reinforcement Learning-based Task Software Caching Update Algorithm}
Up to now, we can find a mutually satisfactory COMO decision for all users (represented by $\bm{\alpha}_t^{*}$) under any given MEC server's caching state $\bm{b}_t$ and user request state $\bm{\mu}_t$ in any time slot. In other words, we can compute the corresponding energy consumption of any caching state $\bm{b}_t$ under any given user request state $\bm{\mu}_t$ since the COMO decision $\bm{\alpha}_t^*$ can be solved by using \textbf{Algorithm 1}. Substitute $\bm{\alpha}_t^{*}$ into the original problem $\mathcal{P}$, the original problem $\mathcal{P}$ can be transformed to the TSCU problem as
\begingroup
\setlength{\abovedisplayskip}{0pt}
\setlength{\belowdisplayskip}{2pt}
\begin{align}
\mathcal{P}_2:~~\min_{\bm{\beta}_t}~~&\mathop {\lim}\limits_{T \to \infty } \frac{1}{T}\sum\nolimits_{t = 1}^T \sum\nolimits_{k \in \mathcal{K}} {{\widehat E}_{k,t}} \label{prob:Pro_TCU}\\
\text{s.~t.~~}& (\text{\ref{cons:one1}}),(\text{\ref{cons:one3}}),(\text{\ref{cons:one4}}),(\text{\ref{cons:one6}}). \notag
\end{align}
\endgroup
\vspace{-0.3cm}
where
\begin{equation}
\setlength{\abovedisplayskip}{3pt}
\setlength{\belowdisplayskip}{3pt}
{\widehat E_{k,t}} = \sum\limits_{f \in \mathcal{F}} \mathbbm{1}(u_k^{(t)} = f)\bigg\{\bigg. \mathbbm{1}(\alpha_{k,t}^* = 0)E_{k,f}^{\text{L}}
 + \mathbbm{1}(\alpha_{k,t}^* \in \mathcal{M})\left( (1 - b_f^{(t)})E_{k,f,t}^{\text{O}} + b_f^{(t)} E_{k,f,t}^{\text{C}} \right) \bigg\}\bigg..
\end{equation}

Knowing that the TSCU decision $\bm{\beta}_t$ depends on the caching state $\bm{b}_t$, it is complex to directly solve $\bm{\beta}_t$. For ease of solving problem $\mathcal{P}_2$, we first solve the optimal caching state $\bm{b}_{t+1}$ in time slot $t+1$, then obtain the caching update decision $\bm{\beta}_{t}$ in slot $t$ based on $b_f^{(t)} + \beta_f^{(t)} = b_f^{(t + 1)}$.
The optimal caching state problem is formulated as
\begingroup
\setlength{\abovedisplayskip}{0pt}
\setlength{\belowdisplayskip}{2pt}
\begin{align}
\widehat{\mathcal{P}_2}:~~\min_{\bm{b}_{t+1}}~~& \mathop {\lim}\limits_{T \to \infty } \frac{1}{T}\sum\nolimits_{t = 1}^T \sum\nolimits_{k \in \mathcal{K}} {{\widehat E}_{k,t}} \label{prob:Pro_TCU1}\\
\text{s.~t.~~}& \sum\nolimits_{f \in \mathcal{F}} b_f^{(t + 1)}{D_f} \le C,\label{cons:TCU1}\tag{\theequation a}\\
&b_f^{(t + 1)} \in \left\{ {0,1} \right\}, \forall f \in \mathcal{F}.\label{cons:TCU2}\tag{\theequation b}
\end{align}
\endgroup
For any time slot ($t+1$), we can solve the optimal caching state $\bm{b}_{t+1}$ when the user request $\bm{\mu}_{t+1}$ is given, e.g., we can solve the energy consumption of all caching state $\bm{b}_{t+1}$ and find the minimum one. However, the caching state $\bm{b}_{t+1}$ is given by the MEC server updates caching space at the end of time slot $t$, and $\bm{\mu}_{t+1}$ is unknown at that time due to the unknown user request transition probabilities.
To tackle this challenge, we apply a DDQN to capture the features of the users' request model and predict the optimal task caching state of time slot ($t+1$) based on the system state of slot $t$. For the purpose of designing the DDQN algorithm, we reformulate problem $\widehat{\mathcal{P}_2}$ as an MDP and elaborate the state, action and reward in the below.
\begin{itemize}
\item State: the state in time slot $t$ is the user request state, i.e., ${S_t} = {\bm{\mu}_t} \in {(F + 1)^K}$.

\item Action: the action in time slot $t$ is the caching state in slot $(t+1)$, i.e., $A_t = \bm{b}_{t + 1} \in {\left\{{0,1} \right\}^F}$.

\item Reward: we define the reward in time slot $t$ as the saving value of energy consumption in time slot $(t+1)$, i.e., $R_{t+1}$. The saving value of energy consumption is defined as the difference between non-caching-based computing cost and caching-based computing cost, i.e., $R_{t+1} = E_{t+1}^{\text{NC}} - E_{t+1}^{\text{C}}$, where
\begin{equation}
\setlength{\abovedisplayskip}{3pt}
\setlength{\belowdisplayskip}{3pt}
E_{t+1}^{\text{NC}} = \sum\limits_{k \in \mathcal{K}}  \sum\limits_{f \in \mathcal{F}} \mathbbm{1}(u_k^{(t+1)} = f)\bigg\{\bigg. \mathbbm{1}(\alpha_{k,t+1}^{\text{NC}} = 0)E_{k,f}^{\text{L}}
+ \mathbbm{1}(\alpha_{k,t+1}^{\text{NC}} \in \mathcal{M})E_{k,f,t+1}^{\text{O}} \bigg\}\bigg.
\end{equation}
is the energy consumption when the MEC server's caching state is empty, i.e., $\bm{b}_{t+1}=[0]_F$,
\vspace{-1cm}
\begin{multline}
E_{t+1}^{\text{C}} \!=\! \sum\nolimits_{k \in \mathcal{K}} \sum\nolimits_{f \in \mathcal{F}} \mathbbm{1}(u_k^{(t+1)} \!=\! f)\bigg\{\bigg. \mathbbm{1}(\alpha_{k,t+1}^{\text{C}} \!=\! 0)E_{k,f}^{\text{L}}\\
 + \mathbbm{1}(\alpha_{k,t+1}^{\text{C}} \!\in\! \mathcal{M})\left( {(1 \!-\! b_f^{(t + 1)})E_{k,f,t+1}^{\text{O}} \!+\! b_f^{(t + 1)}E_{k,f,t+1}^{\text{C}}} \right) \bigg\}\bigg.
\end{multline}
is the energy consumption when the caching state is $\bm{b}_{t+1}$,
where $\alpha_{k,t+1}^{\text{NC}}$ is the COMO decision when the caching space is empty, and $\alpha_{k,t+1}^{\text{C}}$ corresponds to the COMO decision when the caching state is $\bm{b}_{t+1}$. Both $\alpha_{k,t+1}^{\text{NC}}$ and $\alpha_{k,t+1}^{\text{C}}$ can be solved by \textbf{Algorithm 1}.
\end{itemize}

The architecture of the applied DDQN is shown in Fig. \ref{fig:Agent}, which includes two DNNs with same structure: one is the main network, one is the target network.
The DDQN aims to learn the user request model and predict the optimal task software caching state in the next slot based on the user request in the current slot.
Instead of using a large Q table to list all possible states and actions, the applied DDQN in this paper uses a DNN to avoid listing all possible states and actions.
To overcome the high-dimension and complex caching action space resulting from massive tasks with heterogeneous data size and improve learning efficiency, we provide a new design of the DNN,  named state coding and action aggregation (SCAA).
SCAA adopts a dropout mechanism in the input layer to code users' states and a two-layer architecture at the output layer to aggregate caching actions dynamically. Fig. \ref{fig:DNN} shows the architecture of the proposed SCAA-DNN of the DDQN. In the following part, we introduce the SCAA-DNN in detail.
\begin{figure}
  \begin{minipage}[t]{0.5\linewidth}
    \centering
    \setlength{\abovecaptionskip}{0cm}
    \includegraphics[scale=0.28]{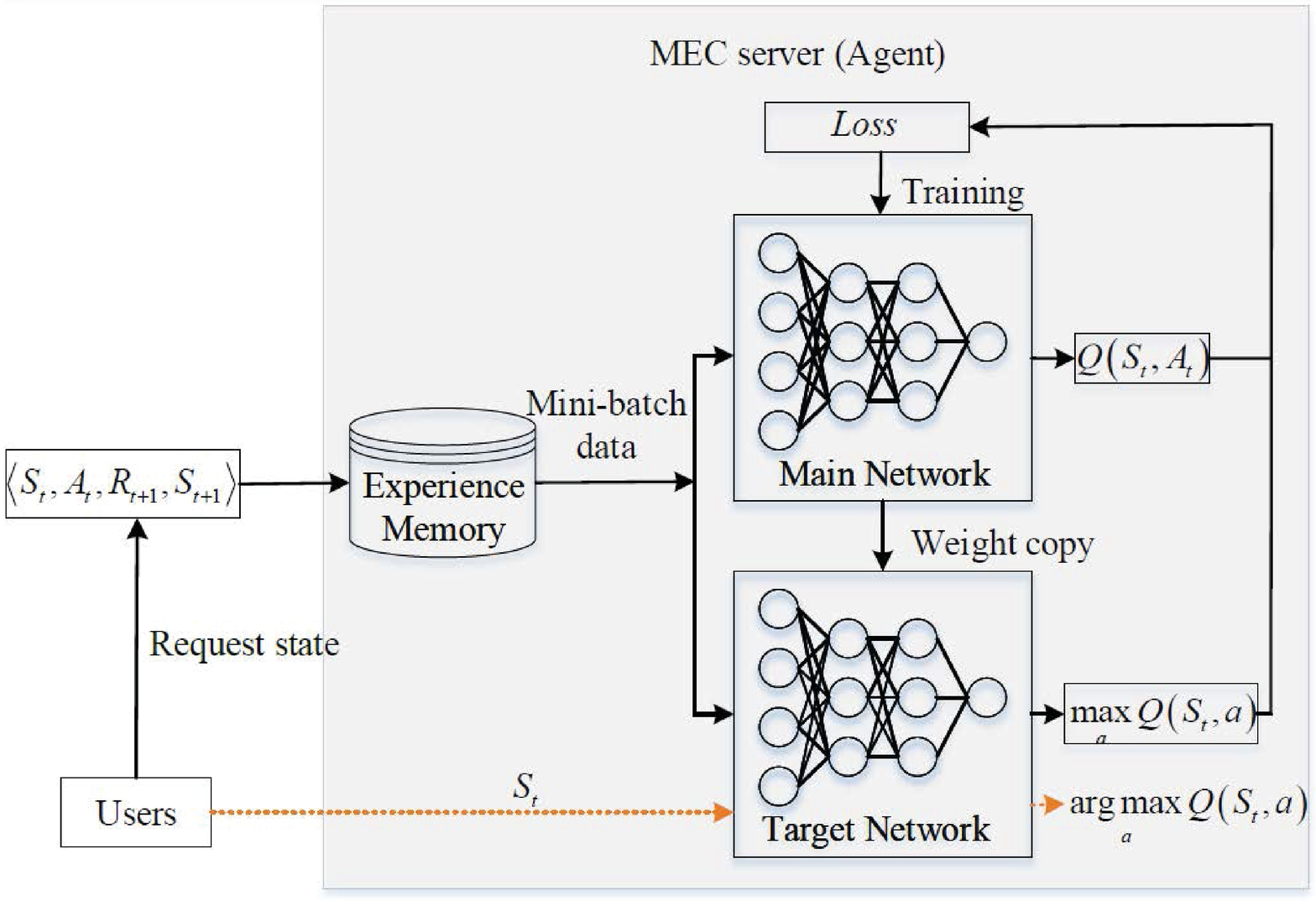}
    \caption{The DDQN training and inference process.}
    \label{fig:Agent}
  \end{minipage}
  \hspace{.15in}
  \begin{minipage}[t]{0.5\linewidth}
    \centering
    \setlength{\abovecaptionskip}{0cm}
    \includegraphics[scale=0.38]{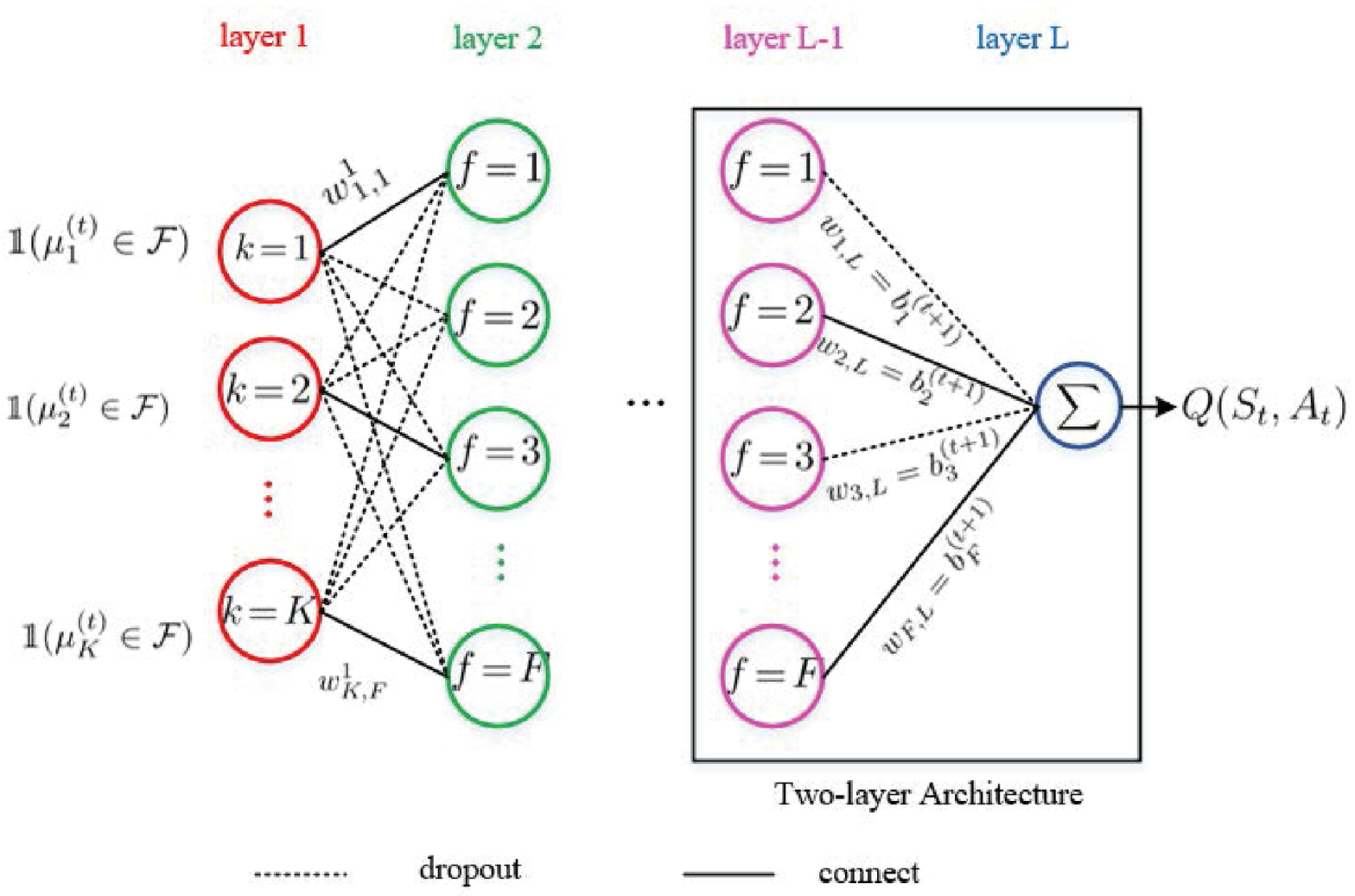}
    \caption{The architecture of proposed SCAA in the DNN of DDQN.}
    \label{fig:DNN}
  \end{minipage}%
  \vspace{-0.8cm}
\end{figure}

In the input of the SCAA-DNN, the users' task request is represented by the task order. For example, $\mu_k^{(t)}=f$ indicates that user $k$ request to execute the $f$-th task in time slot $t$.
The conventional design \cite{9110932} directly uses the state $S_t=\bm{\mu}_t$ as the input variables of the DNN, the tasks' order number will influence the output of the DNN (i.e., the state-action value $Q(S_t,A_t)$). In fact, the order number does not relate to the state-action value $Q(S_t,A_t)$.
In order to eliminate the influence of tasks' order, we use ${\bm{X}_t} = \{ \mathbbm{1}(\mu_k^{(t)} \in \mathcal{F}):k \in \mathcal{K} \}$ as the input of the DNN instead of the state $S_t$. The first layer of the DNN contains $K$ neural cells, and the input of the $k$-th cell is  $\mathbbm{1}(\mu_k^{(t)} \in \mathcal{F})$.
Hence, for clarifying the task demands of users, we define the second layer in the DNN contains $F$ neural cells, in which the $f$-th cell corresponds to the $f$-th task.
We use ${\bm{w}_1} = \{ {w_{k,f}^1: k \in \mathcal{K},f \in \mathcal{F}} \}$ to denote the weights of connections between the first layer and the second layer of the DNN, where $w_{k,f}^1$ denotes the weight of connection between the $k$-th neural cell in the first layer and the $f$-th neural cell in the second layer. The value of $w_{k,f}^1$ is defined as
\vspace{-0.4cm}
\begin{equation}\label{eq:weight1}
w_{k,f}^1 = \left\{ {\begin{array}{*{20}{c}}
{0,}&{\text{if}~~\mu_k^{(t)} \in \mathcal{F}~~\text{and}~~ \mu_k^{(t)} \ne f,}\\
{w_{k,f}^1,}&{\text{otherwise}.}
\end{array}} \right.
\end{equation}
If $\mu_k^{(t)} \in \mathcal{F}$, the connections between the $k$-th neural cell in the first layer and neural cells in the second layers except from the $\mu_k^{(t)}$-th neural cell will be dropout. In other words, the output of the $k$-th neural cell in the first layer only as the input of the $\mu_k^{(t)}$-th neural cell in the second layer.
If $\mu_k^{(t)} = 0$, all the connections between the $k$-th neural cell in the first layer and neural cells in the second layers will be remained, and $\mu_k^{(t)}$ does not affect the inputs of neural cells in the second layer. Such a design implements the user requests state coding in actuality.

In the conventional DDQN \cite{van2016deep}, the number of neural cells in the output layer of the DNN is equal to the number of all possible actions, in which each neural cell corresponds to one action and output the corresponding state-action value, i.e., $Q({S_t},{A_t})$.
However, for the caching problem $\widehat{\mathcal{P}_2}$, it is impractical due to the heterogeneous data size of task software and the large number of tasks. The large number of tasks will produce a large number of possible caching actions. Besides, it is difficult to list all the possible actions due to the heterogeneous size of task software.
For example, we assume that the MEC server can cache 10 task software with the same data size, and the task library has 50 tasks. The MEC server will have $C_{50}^{10} = 1.0272 \times 10^{10}$ possible actions. If the data sizes of these tasks' software are different, it is more complex to combine all available caching actions.
To tackle this challenge, we use a two-layer architecture (TLA) as the output layer of the SCAA-DNN, shown in Fig. \ref{fig:DNN}. The first layer in the TLA contains $F$ neural cells, in which the $f$-th neural cell corresponds to task $f$. Let $O = \left( {{O_1},{O_2}, \cdots ,{O_F}} \right)$ denote the output of the first layer of the TLA. Intuitively, $O_f$ represents the part of state-action value of caching the task $f$-th software. The last layer of the TLA just has one neural cell which does not have the activation unit and outputs the weighted sum of all input variables.
We use $\bm{w}_L = (w_{1,L}, \cdots ,w_{f,L}, \cdots ,w_{F,L})$ to denote the weights of connections between the first layer and the last layer in the TLA, where $w_{f,L}$ is the weight of the connection between the $f$-th cell in the first layer and the last layer in the TLA. To identify the state-action value of a specific action $A_t=\bm{b}_{t+1}$, we assign the value of $A_t$ to $\bm{w}_L$, i.e., $w_{f,L} = b_f^{(t+1)},\forall f \in \mathcal{F}$.
Then, the DNN will output the predicted state-action value, i.e., $ Q(S_t, A_t) =\sum\nolimits_{f \in \mathcal{F}} { b_f^{(t+1)} O_f}$.

\begin{rem}\label{rem:DDQN_complexity}
In practical caching scenarios, the large number of tasks in the library may produce a high-dimension action space and complex network structure in the DDQN because the caching action is a combination of caching some task software. It may result in many neural cells in the output layer of the DNN used in the DDQN, hindering the convergence of the DDQN.
Using the proposed TLA, the complexity of the used neural network in the DDQN is significantly reduced, thus improving the convergence speed of the DDQN. Note that such a design also can be used in other scenarios with high-dimension combined-action space.
\end{rem}

About the training phase, the MEC server caches task software based on the $\varepsilon$-greedy policy \cite{sutton2018reinforcement} at the end of time slot $t$, where the MEC server randomly cached task software with probability $\varepsilon$ or caches task software based on $A_t^*=\argmax\nolimits_{a} Q(S_t, a)$ with probability $(1- \epsilon)$.
At the beginning of slot ($t+1$), the users will generate task computing requests $\bm{\mu}_{t+1}$ and find the COMO decisions $\bm{\alpha}_{t+1}^*$ through \textbf{Algorithm 1} based on the caching state $\bm{b}_{t+1}$ and user request $\bm{\mu}_{t+1}$.
Then, the users accomplish their tasks based on $\bm{\alpha}_{t+1}^*$ and result in energy consumption, i.e., $E_{t+1}^{\text{C}}$. To estimate the reward of the caching action $A_t=\bm{b}_{t+1}$, we set the caching state as empty (i.e., $\bm{b}_{t+1}=\bm{0}$) and obtain the corresponding energy consumption, i.e., $E_{t + 1}^{\text{NC}}$. The user request state $S_t$ in time slot $t$, the action $A_t$, the reward $R_{t+1} = E_{t+1}^{{\text{NC}}} - E_{t+1}^{\text{C}}$, and the state $S_{t+1}$ in the next time slot will be stored in the experience memory and used as the training data for the DDQN. Then, the DDQN samples a batch of data from the experience memory as the training data, each data is in the form of $\left\langle {S_t, A_t, R_{t + 1}, S_{t+1}} \right\rangle$.

Firstly, the DDQN assigns $A_t$ to $\bm{w}_L$ of the evaluation DNN, i.e., $w_{f,L}=b_f^{(t+1)}, \forall f \in \mathcal{F}$. Then, the DDQN assigns values to the weights between the first layer and second layer of the evaluation DNN based on Eq. (\ref{eq:weight1}) and input ${\bm{X}_t} = \{ { \mathbbm{1}(\mu_k^{(t)} \in \mathcal{F}):k \in \mathcal{K}} \}$. Next, the evaluation DNN accomplishes forward process and obtains the predicted state-action value, i.e., $Q(S_t, A_t)$. The training process should make $Q(S_t, A_t)$ approximate the expected state-action value as
\begin{equation}
\setlength{\abovedisplayskip}{3pt}
\setlength{\belowdisplayskip}{3pt}
\bar{Q}(S_t, A_t) = R_{t+1} + \gamma \max_a Q(S_{t+1}, a)
\end{equation}
where $\gamma \in (0,1)$ is discount factor.
For computing the expected state-action value, we use the target DNN in the DDQN to inference the value of $\max_a Q(S_{t+1}, a)$.
To make the learning process more stable, we use the Huber function \cite{agarwal2020optimistic} to quantify the loss instead of the square error function. The loss function is defined as follows.
\begin{equation}
\setlength{\abovedisplayskip}{3pt}
\setlength{\belowdisplayskip}{3pt}
Loss = \left\{ {\begin{array}{*{20}{c}}
{\frac{1}{2}{{(Q(S_t, A_t) - \bar{Q}(S_t, A_t))}^2},}&{\left| {Q(S_t, A_t) - \bar{Q}(S_t, A_t)} \right| < 1}\\
{\left| {Q(S_t, A_t) - \bar{Q}(S_t, A_t)} \right| - \frac{1}{2}},&{\text{otherwise}}.
\end{array}} \right.
\end{equation}
Once the loss function value is calculated, we can train the evaluation DNN by using backward algorithm \cite{goodfellow2016deep}. The detailed steps of the DDQN training algorithm are listed in \textbf{Algorithm 2}.
\begin{algorithm}
\algsetup{linenosize=} \small
\caption{The Training Algorithm for DDQN}
\begin{spacing}{1.3}
\begin{algorithmic}[1]
\STATE Initialize replay memory with capacity $E$, the weight copy frequency $g$
\STATE Initialize the evaluation DNN with random weights $\theta$ and copy $\theta$ to the target DNN
\FOR{time slot $t =1:T$}
    \STATE With probability $\epsilon$ select a random caching state $A_t$ otherwise select $A_t=\argmax\nolimits_{a} Q(S_t,a)$ as the caching state in slot $t+1$
    \STATE Use $A_t$ as the caching state of the MEC server in time slot ($t+1$) and compute the reward $R_{t+1}$
    \STATE Store transition $S_t, A_t, R_{t+1}, S_{t+1}$ in experience memory
    \STATE Sample random mini-batch of transitions $S_t, A_t, R_{t+1}, S_{t+1}$ from experience memory
    \STATE Assign values to the weights between the first layer and second layer based on Eq. (\ref{eq:weight1}).
    \STATE Assign $\bm{b}_{t+1}$ to the weights of the TLA in the evaluation DNN.
    \STATE Input ${\bm{X}_t} = \{ { \mathbbm{1}(\mu_k^{(t)} \in \mathcal{F}):k \in \mathcal{K}} \}$ to the evaluation DNN and obtain $Q(S_t,A_t)$
    \STATE Perform a gradient descend step on loss function with respect to the DNN parameters
    \STATE Update the target DNN every $g$ slots
\ENDFOR
\end{algorithmic}
\end{spacing}
\end{algorithm}

In the DDQN inference phase, we first assign values to the weights between the first layer and second layer based on Eq. (\ref{eq:weight1}). Then, we input ${\bm{X}_t} = \{ \mathbbm{1}(\mu_k^{(t)} \in \mathcal{F}):k \in \mathcal{K} \}$ to DNN and forwards to the first layer of TLA and output $\bm{O} = \left( {{O_1},{O_2}, \cdots ,{O_F}} \right)$. Finally, we need find the optimal caching state in time slot ($t+1$) (i.e., $\mathop {\argmax }\nolimits_{\bm{b}_{t+1}} Q(\bm{\mu}_t,\bm{b}_{t+1})$). We formulate the optimal caching state problem as follows.
\begingroup
\setlength{\abovedisplayskip}{0pt}
\setlength{\belowdisplayskip}{2pt}
\begin{align}
\widetilde{\mathcal{P}_2}:~~\max_{\bm{b}_{t+1}}~~& \sum\nolimits_{f \in \mathcal{F}} {b_f^{(t+1)} O_f} \label{prob:Pro_maxAction}\\
\text{s.~t.~~}& \sum\nolimits_{f \in \mathcal{F}}  b_f^{(t+1)} D_f \le C,\label{cons:MA1}\tag{\theequation a}\\
&b_f^{(t+1)} \in \left\{ {0,1} \right\}.\label{cons:MA2}\tag{\theequation b}
\end{align}
\endgroup
Problem $\widetilde{\mathcal{P}_2}$ is a typical Knapsack problem \cite{garey1979computers}. Below we introduce a recursive function to derive the optimal solution. For ease of presentation, we first define a $F \times C$ matrix $\Xi$, in which $\Xi(f,c)$ represents the optimal solution under the first $f$ tasks using a cache size of $c$. The value of $\Xi(f,c)$ is given by the following recursive function.
\vspace{-0.3cm}
\begin{equation}
\setlength{\abovedisplayskip}{3pt}
\setlength{\belowdisplayskip}{3pt}
\Xi(f,c) = \mathop {\max}\nolimits_{b_f^{(t+1)}} ( \Xi(f - 1,c - b_f^{(t+1)} D_f) + b_f^{(t+1)} O_f ).
\end{equation}

Through the above recursive function, the optimal solution of problem $\widetilde{\mathcal{P}_2}$ can be derived by the argument of $\Xi(F,C)$. For clarity, we conclude the detailed steps of solving optimal caching state in \textbf{Algorithm 3} whose time complexity is $\mathcal{O}(2FC+F)$.
\begin{algorithm}[t]\small
\caption{Algorithm for Solving the Optimal Action}
\hspace*{0.02in} {\bf Input:}
$O_t, \{D_f: f \in \mathcal{F}\}$\\
\hspace*{0.02in} {\bf Output:}
 The optimal caching state $\bm{b}_{t+1}$
\begin{algorithmic}[1]
\STATE $\bm{b}_{t+1}=[0]_{F}, \Xi= \left[ 0 \right]_{F \times C}, \Xi_r= \left[ 0 \right]_{F \times C}$;
\FOR{each $f \in [1,F]$}
    \IF{$f < F$}
        \FOR{each $c \in [1,C]$}
            \IF{$f==1$}
                \STATE $\Xi_r(f,c)= \mathbbm{1}(D_f < c)$
                \STATE $\Xi(f,c)=\Xi_r(f,c) O_f$
            \ELSE
                \STATE ${\Xi_r}(f,c) = \argmax\limits_{a \in \left\{0,1 \right\}} ( \Xi(f - 1,c - a D_f) + a O_f )$
                \STATE $\Xi(f,c) = {\Xi_r}(f,c){O_f} + \Xi(f - 1,c - {\Xi_r}(f,c){D_f})$
            \ENDIF
        \ENDFOR
    \ELSE
        \STATE ${\Xi_r}(F,C) = \argmax\nolimits_{a \in \left\{ {0,1} \right\}} ( \Xi(F - 1,C - a D_F) + a O_F )$
        \STATE $\Xi(F,C) = {\Xi_r}(F,C){O_F} + \Xi(F - 1,C - {\Xi_r}(F,C){D_F})$
    \ENDIF
\ENDFOR

\STATE $\bm{b}_{t+1}(F) = \Xi_r(F,C)$
\FOR{each $f =F-1:-1:1$}
    \STATE $\bm{b}_{t+1}(f)={\Xi_r}(f,C-\sum_{f+1\leq j\leq F}\bm{b}_{t+1}(j)*D_j)$
\ENDFOR
\RETURN $\bm{b}_{t+1}$
\end{algorithmic}
\end{algorithm}

Once the optimal caching state $\bm{b}_{t+1}$ in time slot $t+1$ is derived, the MEC server can calculate the optimal TSCU policy in time slot $t$, i.e., $ \beta_f^{(t)} = b_f^{(t + 1)}- b_f^{(t)}$. Then, the MEC server can update its cache space and assist the COMO in time slot ($t+1$). For clarity, we conclude the detailed steps of the DDQN inference phase in \textbf{Algorithm 4}.
In addition, for ease of understanding, Fig. \ref{fig:alg_con} illustrates the connections between all algorithms and the physical system model.
\begin{algorithm}
\algsetup{linenosize=} \small
\caption{The Inference Algorithm of DDQN}
\begin{spacing}{1.3}
\begin{algorithmic}[1]
\STATE Assign values to the weights between the first layer and second layer based on Eq. (\ref{eq:weight1}).
\STATE Input ${X_t} = \{ { \mathbbm{1}(\mu_k^{(t)} \in \mathcal{F}):k \in \mathcal{K}} \}$ to the first layer of the DNN, then the DNN forwards to the first layer of TLA and output $\bm{O} = \left( {{O_1},{O_2}, \cdots ,{O_F}} \right)$
\STATE Solve the optimal caching state in the next time slot using \textbf{Algorithm 3}
\STATE Calculate the optimal caching update policy based on $\beta_f^{(t)} = b_f^{(t + 1)}- b_f^{(t)}, \forall f \in \mathcal{F}$
\end{algorithmic}
\end{spacing}
\end{algorithm}

\begin{figure}[t!]
  \centering
  \setlength{\abovecaptionskip}{0.cm}
  \includegraphics[width=0.9\textwidth]{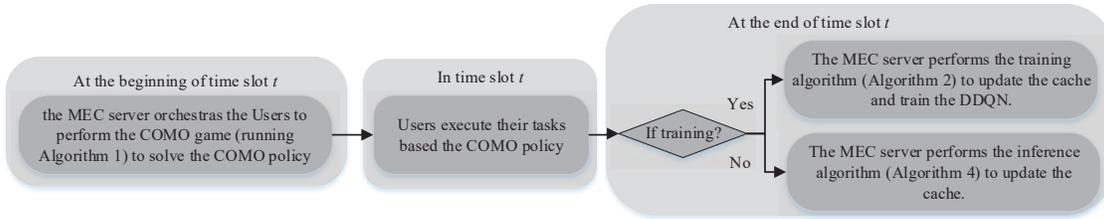}\\
  \caption{An illustration for connections between the  proposed algorithms and the system model.}
  \label{fig:alg_con}
  \vspace{-0.8cm}
\end{figure}
\vspace{-0.8cm}
\section{Simulation Results}\label{sec:simu}
This section evaluates the proposed dynamic TSCU-based COMO scheme by comparing its performances with the following baseline schemes. Note that these baselines for caching updates do not include the COMO policy. For fairness, we add the COMO policy proposed in this work to these baselines for forming TSCU assisted COMO schemes. Moreover, we use the COMO policy proposed in this work as a baseline for illustrating the advantages of TSCU.
\begin{itemize}
  \item The least recently used caching-based MEC (LRU-MEC) updates task software caching based on LRU policy \cite{8362880}, in which the MEC server keeps the most recently requested task software in the MEC server cache memory. When the cache storage is full, the cached task software, which is requested least recently, will be replaced by the new task software.

  \item The least frequently used caching-based MEC (LFU-MEC) updates task software caching based on LFU policy \cite{8362880}, in which the MEC server caches the task software with highest request count which is calculated by the request information of past time slots. When the cache storage is full, the cached task software, which is requested the least many times, will be replaced by the new task software.

  \item The first in first out-based MEC (FIFO-MEC) update task software caching according to FIFO policy \cite{8778670}.

  \item The local most popular caching-based MEC (LMP-MEC) updates the cache based on LMP algorithm \cite{9110932}, which predicts the next request based on both long-term file popularity and short-term temporal correlations in request sequences.

  \item MEC offloading: The MEC offloading scheme utilizes the proposed potential game-based COMO algorithm to decide the executive method of users' tasks under the empty task software caching state of the MEC server. It only has two ways of task computing, i.e., local computing and non-caching based COMO.

\end{itemize}

In the simulations, the proposed scheme and benchmark schemes are implemented using Python and Pytorch. It is assumed that $K$ users are randomly distributed over a $200$m$\times 200$m single cell, and the BS is sited in the cell's center.
Similar with \cite{rappaport1996wireless, 7307234}, the channel gain is modeled as $h_{k,t} = \rho_k(t) d_k^{- n}$ where $d_k$ is the distance between user $k$ and the BS, $\rho_k(t) \sim \text{Exp}(1)$ is exponentially distributed with unit mean, which represents the small-scale fading channel power gain from user $k$ to the MEC server in slot $t$, and $n$ is the path loss factor.
According to the realistic measurements in \cite{miettinen2010energy}, we set the energy coefficient $\zeta$ as $5 \times {10^{-27}}$.
The input parameters data size of each task, i.e., $I_f$, is uniform randomly selected in $[1,I_{\max}]$ Megabytes.
The software data size of each task, i.e., $D_f$, is uniform randomly selected in $[1,D_{\max}]$ Gigabytes.
The required CPU cycles for computing task $k$, i.e., $S_f$, is randomly selected in $[1,S_{\max}]$ Gigacycles.
The parameters chosen in the simulation are based on the parameter setting of a typical MEC network \cite{7307234,9110932,8418400}.
Unless otherwise stated, the primary simulation environment settings are summarized in Table II.

In terms of the user task request $\mu_k^{(t)}$, we use $\Pr [\mu_k^{(t + 1)} = j | {\mu_k^{(t)} = i}]$ to denote the transition probability from task $i$  to $j$ ($i, j \in \overline{\mathcal{F}}$) of user $k$.
Similar to \cite{9110932} and \cite{8418400}, we assume that all users' request transition probabilities follow the same request transition model as follows.
\begingroup
\setlength{\abovedisplayskip}{0pt}
\setlength{\belowdisplayskip}{2pt}
\begin{equation}
\setlength{\abovedisplayskip}{3pt}
\setlength{\belowdisplayskip}{1pt}
\Pr [\mu_k^{(t + 1)} \!=\! j | {\mu_k^{(t)} = i} ]
= \left\{ {\begin{array}{*{20}{c}}
{R},\!&\!{i \in \bar{\mathcal{F}},j = 0},\\
{(1 - R)\frac{1/j^{\delta}}{\sum\nolimits_{m = 1}^F 1/m^{\delta} }},\!&\!{i = 0,j \in \mathcal{F}},\\
{(1 - R)\frac{1}{N}},\!&\! j = (i + q)\text{mod}(F + 1), i \in \mathcal{F},q \in \{ 1,2, \cdots,N\},\\
{0,}\!&\!{\text{otherwise}}.
\end{array}} \right.
\end{equation}
\endgroup
$\Pr [ \mu_k^{(t + 1)} \!=\! j | {\mu_k^{(t)} = i}]$ is parameterized by $\left\langle{R,\delta,N} \right\rangle$. Specifically, $R$ is the transition probability of requesting nothing given any task request at the current time slot. The transition probability of any task $f \in \mathcal{F}$ under no current file request is modeled as a Zipf distribution which parameterized by $\delta$. For any task $i \in \mathcal{F}$, we assign a set of $N$ neighboring tasks, i.e., ${\mathcal{N}} = \left\{ {f \in \mathcal{F},f = (i + n)\mod(F + 1),n = 1,2, \cdots N} \right\}$. Then, the transition probability of requesting any task $f \in \mathcal{N}$ under the current task request $i \in \mathcal{F}$ is modeled as the uniform distribution. The transition probability of requesting any task $f \notin \mathcal{N}$ under the current task request $i \in \mathcal{F}$ is zero.
It is worth mentioning we provide the transition probability in the simulation parts to establish the environment. It does not mean the proposed solution relies on the known transition model. In fact, the proposed solution is a model-free approach. In the following results, we alter the transition probability parameters to verify that the proposed solution has the ability to handle problems with different transition probabilities.
\begin{table}[ht]\label{tab:simu_para}
\vspace{-0.6cm}
\small
\caption{Simulation Settings}
\vspace{-0.3cm}
\begin{tabular}{|p{5.5cm}|p{1.5cm}|p{5.5cm}|p{1.5cm}|}
\hline
Parameter & Value & Parameter & Value\\
\hline
User number: $K$    & 20 &
Task number: $F$    & 50 \\
\hline
Number of time slots: $T$    & 2000 &
Wireless transmission bandwidth: $B$    & 30 MHz \\
\hline
Transmission power of devices: $p_k$    & 0.5 W &
White Gaussian noise variance: $\sigma^2$ & $2 \times {10^{-13}}$ \\
\hline
CPU capability of user $k$: $f_k$ & 1 GHz  &
CPU capability of the MEC server: $f_{\text{C}}$ & 20 GHz \\
\hline
Cache size of the MEC server: $C$ & 2 GB &
The number of channels: $M$ & 10 \\
\hline
Path loss factor: $n$ & $4$ &
$I_{\max}$ & 5 \\
\hline
$D_{\max}$ & 5 &
$S_{\max}$ & 5 \\
\hline
Learning rate of DNN & 0.0001 &
Experience replay memory size: $E$ & 1000 \\
\hline
Batch size & 8 & Discount factor: $\gamma$ & 0.9\\
\hline
$\tau$ & 5ms &
& \\
\hline
\end{tabular}
\vspace{-0.6cm}
\end{table}

In Fig. \ref{fig:game_conver}, the black solid curve represents the reduced energy consumption per training slot of the proposed TSCU-based COMO scheme. The black dash line represents
the counterpart with conventional way that uses the user request $\bm{\mu}_t$ as the input of the DNN, and all weights between the first and the second layer are connected.
These two curves are plotted using the moving average with a window equal 20. The blue dash curve shows the dynamics of the system-wide energy consumption in one slot with the empty storage status of the MEC server. We can see that the potential game-based COMO algorithm rapidly converge to a stable point, i.e., the NE of the multi-user COMO game.
Moreover, the reduced energy consumption (black curve) increases as the training slots increase and reaches the maximum reduction value when the learning process becomes stable. It is valuable to note that the proposed scheme can rapidly converge to the maximum reduction value point (less than 1000 slots). Most existing DRL-based caching works usually consume more than $10^4$ training slots, like \cite{8778670,9110932}. Compared with directly inputting users' request state to the DNN, the proposed SCAA approach is able to reduce the learning complexity and accelerate the convergence of the DDQN.

In Fig. \ref{fig:cache_size_EE}, we show that the impact of the MEC server's cache size on the average energy consumption over each time slot of the proposed scheme and the five baselines.
We can see that all schemes' average energy consumption over each time slot, except the MEC offloading scheme, is reduced with the increase of cache size. This reduction is because the larger cache size allows the storage of more task software. Thus, the requested tasks will have a higher hit rate at the MEC server, which means that more users can execute their tasks through a lower-cost method, i.e., caching-based COMO.
When the cache size is 0, the MEC server cannot cache any task software, and all schemes only can execute tasks through non-caching based COMO or local computing. There is no distinction between these schemes in this case. When the cache size is big enough to cache all the task software (over 18GB), all schemes have the same performance.
In this case, the MEC serve can cache all task software in the task library. Thus, the users can execute their tasks through local computing or caching-based COMO, and there is also no difference between these schemes. However, in practical systems, the cache size of the MEC server is limited and usually cannot cache all the task software.
Specifically, when the cache size is 8GB, the proposed scheme save around 39\% energy than LMP-MEC scheme.
\begin{figure}
  \begin{minipage}[t]{0.5\linewidth}
    \centering
    \setlength{\abovecaptionskip}{0.cm}
    \includegraphics[width=1\textwidth]{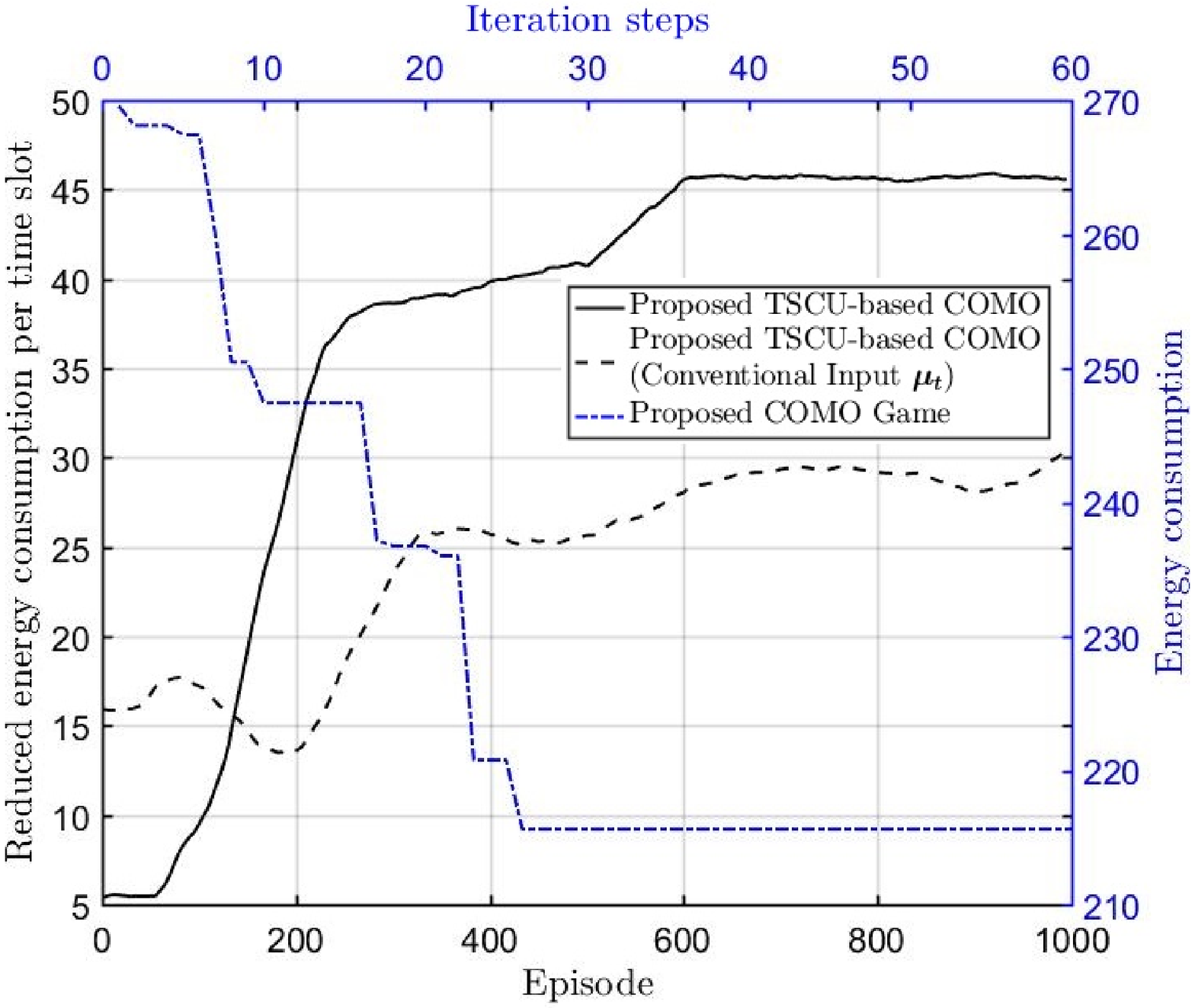}
    \caption{The energy consumption of users with respect to iteration steps under the empty caching state of the MEC server.}
    \label{fig:game_conver}
  \end{minipage}%
  \hspace{.15in}
  \begin{minipage}[t]{0.5\linewidth}
    \centering
    \setlength{\abovecaptionskip}{0.cm}
    \includegraphics[width=1\textwidth]{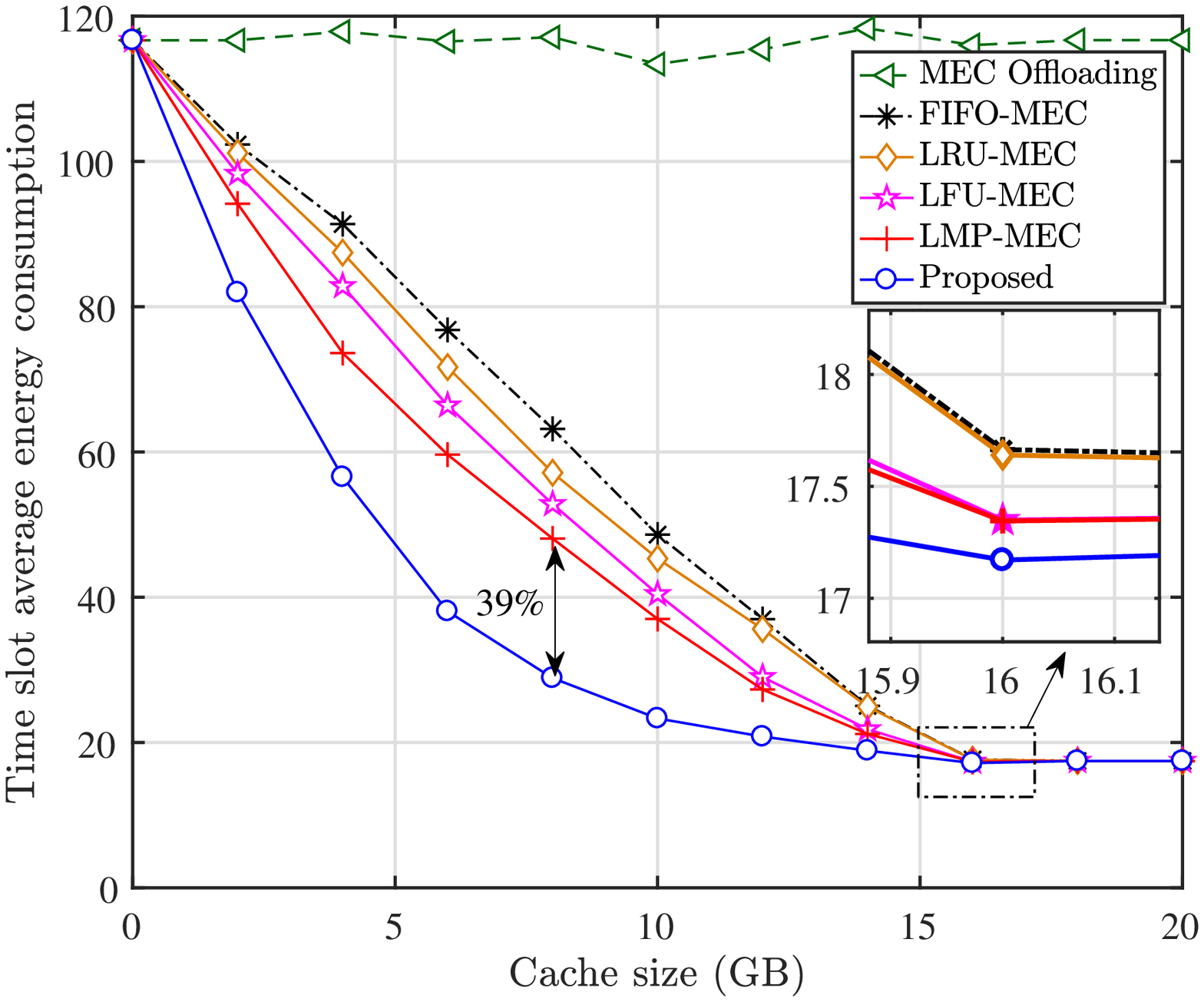}
    \caption{Comparison of the average energy consumption over each time slot against different cache size of the MEC server.}
    \label{fig:cache_size_EE}
  \end{minipage}
\vspace{-0.8cm}
\end{figure}

Fig. \ref{fig:TaskNum_EE} plots the average energy consumption over each time slot of the six schemes versus the number of tasks in the task library.
We can observe that the average energy consumption over each time slot of the caching-based schemes (i.e., LRU-MEC, LFU-MEC, FIFO-MEC, LMP-MEC, and the proposed scheme) increased with the increase of task number. The range of users' task requests will be more expansive with the rise of task number, which may decrease the prediction accuracy of the task software caching schemes and further decrease the reusable of the cached task software.
In addition, it also can be observed that the proposed scheme outperforms the other schemes. When the task number is 10, the proposed scheme can save up to 62\% of energy than the best baseline (LMP-MEC). This benefit comes from the more accurate prediction of users' task demand and the learned knowledge of computing energy consumption about different users.

Fig. \ref{fig:UserNum_delata_EE} shows that how the average energy consumption over each time slot varies with the number of users under different environmental parameters $\delta$. Compared with the best baseline scheme (LMP-MEC), the proposed scheme achieves the lower average energy consumption over each time slot across all user number configurations. Moreover, it is observed that the average energy consumption over each time slot of the two schemes keeps decreasing with the increase of $\delta$. In fact, as $\delta$ increases, most of the user requests concentrate on a few tasks, and the remaining tasks in the library have a very low probability of being requested. Thus, a large $\delta$ is able to improve the prediction accuracy of the two task software caching schemes, and the cached task software has a higher probability of being used. Besides, the proposed scheme saves over 25\% of energy when the user number exceeds 50 compared to the LMP-MEC scheme.
\begin{figure}
  \begin{minipage}[t]{0.5\linewidth}
    \centering
    \setlength{\abovecaptionskip}{0.cm}
    \includegraphics[width=1\textwidth]{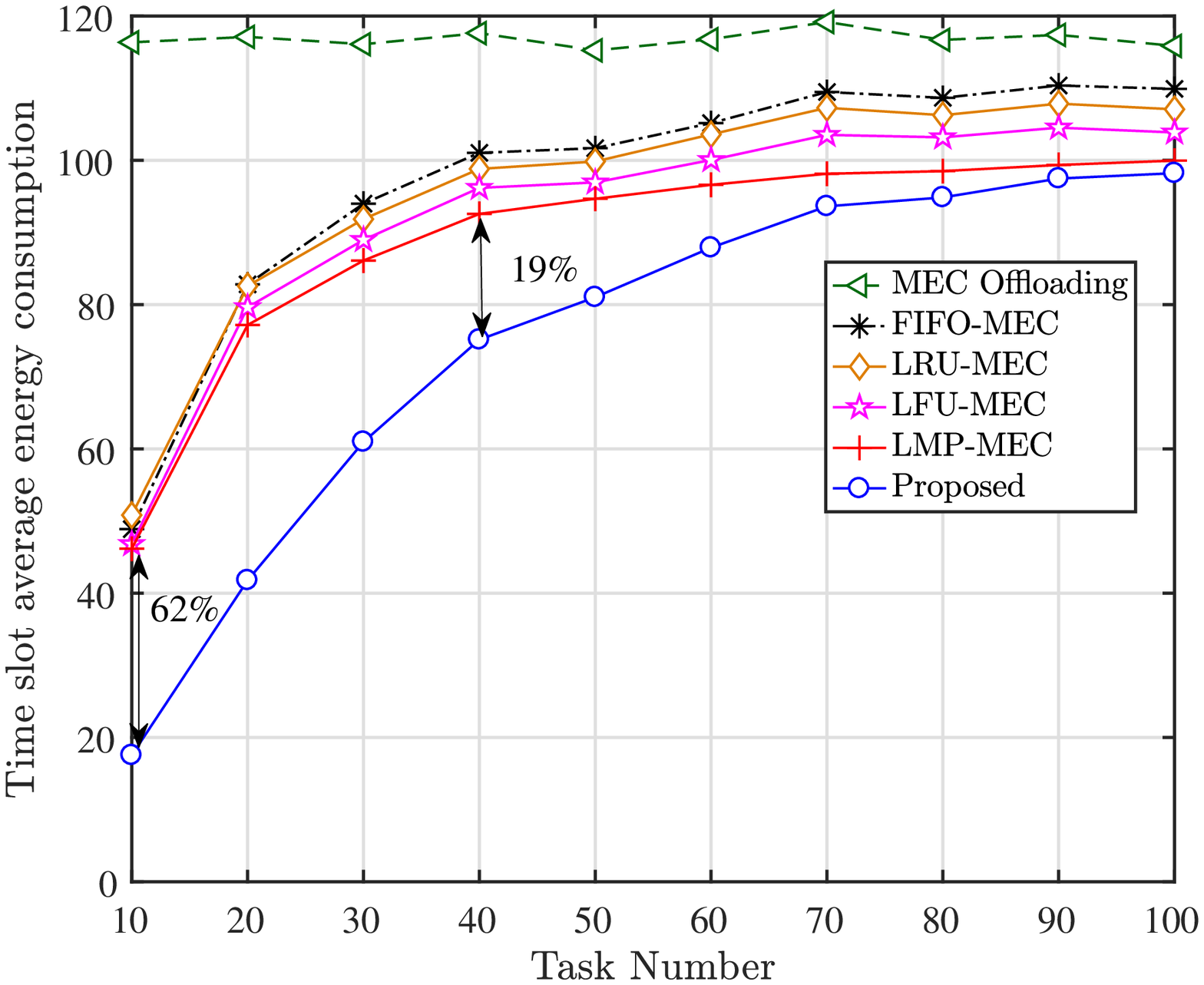}
    \caption{Comparison of the average energy consumption over each time slot against different task number.}
    \label{fig:TaskNum_EE}
  \end{minipage}%
  \hspace{.15in}
  \begin{minipage}[t]{0.5\linewidth}
    \centering
    \setlength{\abovecaptionskip}{0.cm}
    \includegraphics[width=1\textwidth]{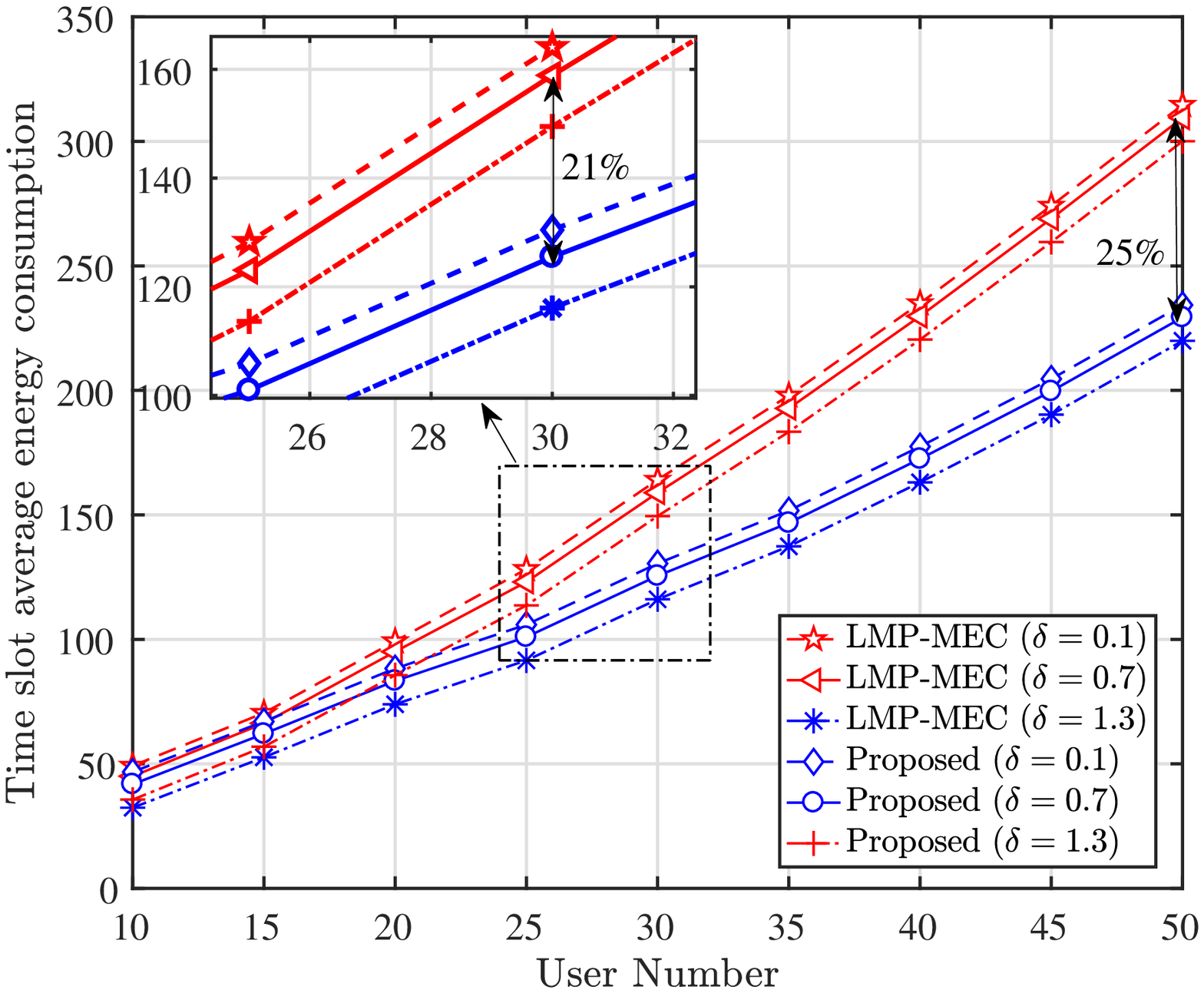}
    \caption{Comparison of the average energy consumption over each time slot against different user number.}
    \label{fig:UserNum_delata_EE}
  \end{minipage}
\vspace{-0.8cm}
\end{figure}

Fig. \ref{fig:Smax_R_EE} plots the average energy consumption over each time slot of the proposed and LMP-MEC scheme. We can see that the average energy consumption over each time slot of both the proposed and LMP-MEC scheme keeps increasing along with the increase of $S_{\text{max}}$.
Using the LMP-MEC scheme as the baseline, the proposed scheme reduces energy consumption by 11.5\% to 22\% across the parameter setting of $S_{\text{max}}$.
The reason is that the growth of $S_{\text{max}}$ will increase the average computation load of tasks, leading to the increases of the local computing energy consumption and the execution delay of the offloaded tasks. The rise of execution delay at the MEC server is likely to reduce the number of offloaded tasks, inducing the average energy consumption growth over each time slot for both schemes. Besides, we can observe that the average energy consumption over each time slot of both schemes decreased with the increase of $R$. The number of users who request to execute tasks will decrease with the rise of $R$. That is to say, the total number of tasks executed in a slot is likely to decline with the increase of $R$, resulting in the growth of average energy consumption.

We reveal the impact of the parameter $D_{\text{max}}$ and $N$ on the average energy consumption over each time slot in Fig. \ref{fig:Dmax_N_EE}. We can see that the average energy consumption over each time slot of the proposed schemes keeps increasing along with the increase of $D_{\text{max}}$. This phenomenon results from that the growth of $D_{\text{ max }}$ will increase the average size of the tasks' software, reducing the number of task software that are cached at the MEC server and increasing the transmission delay and energy consumption of COMO. As the varying of $D_{\text{ max }}$, the proposed scheme is able to save about 12\%-16\% energy compared with the best baseline, LMP-MEC.
Besides, the average energy consumption over each time slot of the proposed scheme increases along with $N$. The reason is that the users' task request range will be more expansive with the increase of $N$, which will reduce the prediction accuracy of the task software caching schemes and further increase the average energy consumption over each time slot. Moreover, the gap between $N=3$ and $N=5$ is larger than the gap between $N=5$ and $N=10$. When $N$ increases to a large number (around 5), every user has the same probability of requesting five tasks. The tasks that all users may request is likely to cover the task library, and the request probability of each task are approximate. In this case, the prediction accuracy may converge to a stable point. Thus, the increment of energy consumption is small with the increase of $N$. In fact, when $N$ increase to a large value, the average energy consumption of all task software caching schemes will keep stable.

\begin{figure}
  \begin{minipage}[t]{0.5\linewidth}
    \centering
    \setlength{\abovecaptionskip}{0.cm}
    \includegraphics[width=1\textwidth]{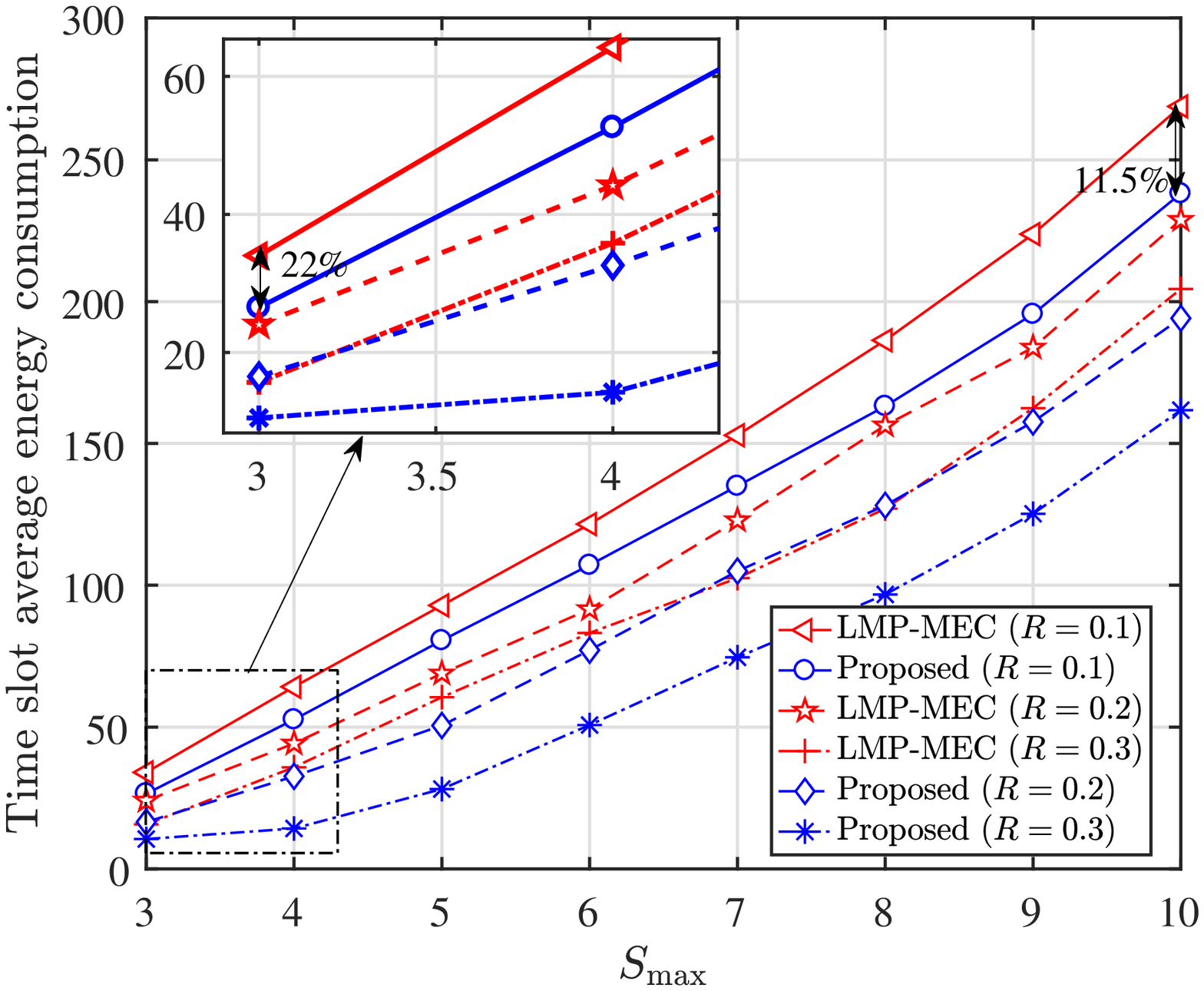}
    \caption{Comparison of the average energy consumption over each time slot against users' request transition probability parameter $R$.}
    \label{fig:Smax_R_EE}
  \end{minipage}%
  \hspace{.15in}
  \begin{minipage}[t]{0.5\linewidth}
    \centering
    \setlength{\abovecaptionskip}{0.cm}
    \includegraphics[width=1\textwidth]{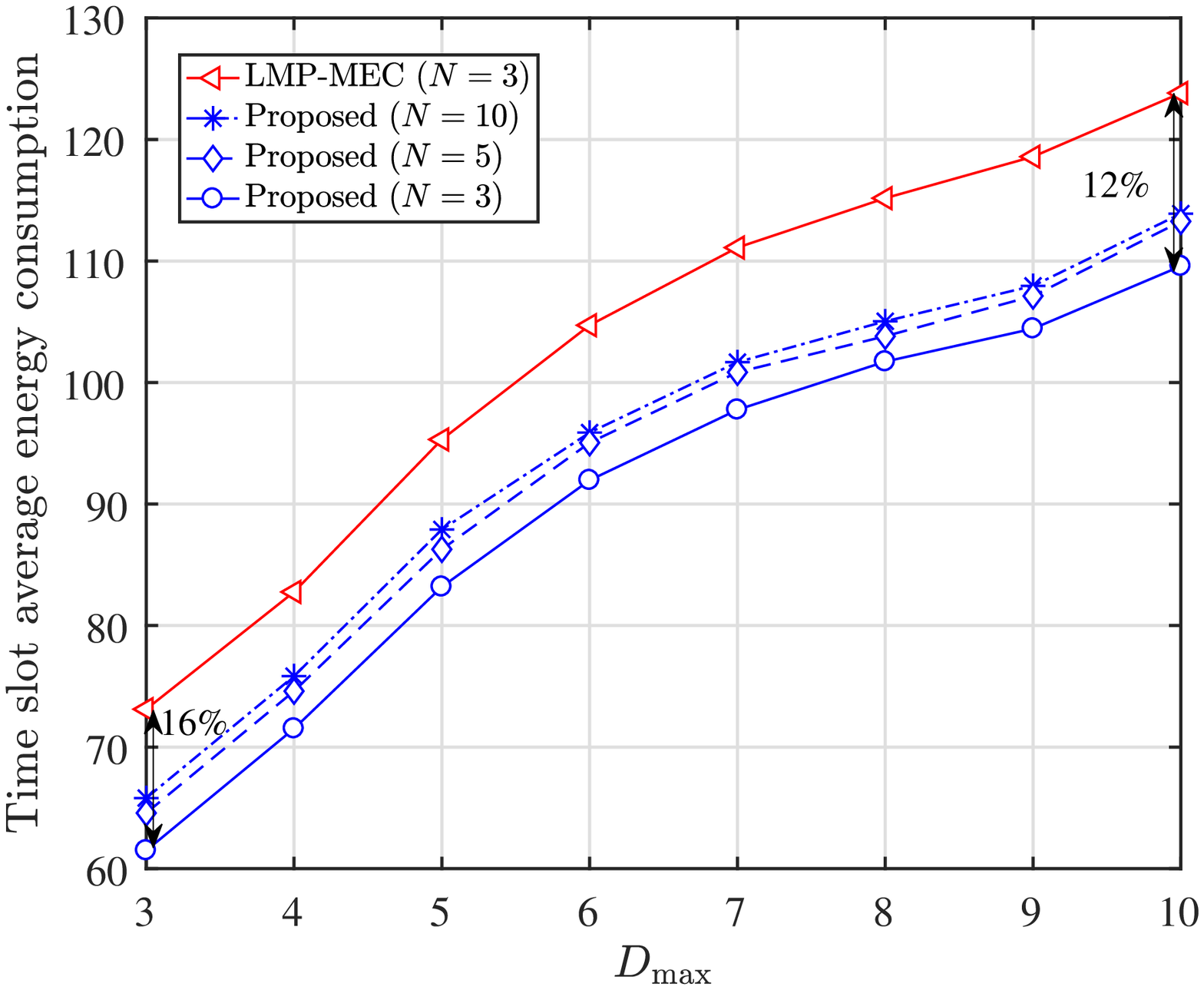}
    \caption{Comparison of the average energy consumption over each time slot against users' request transition probability parameter $N$.}
    \label{fig:Dmax_N_EE}
  \end{minipage}
\vspace{-0.8cm}
\end{figure}
\vspace{-0.8cm}
\section{Conclusion}\label{sec:conclu}
In this paper, we have investigated a joint TSCU and COMO problem in a dynamic multi-user MEC network to minimize the users' task execution energy consumption while satisfying the task execution delay constraint. Through detailed analysis, we have proposed to solve the problem through two stages.
Firstly, we reformulated the COMO problem as a multi-user COMO game and proposed a decentralized COMO algorithm to obtain its NE solution under any task software caching state. Then, we developed a DDQN-based TSCU algorithm to solve the optimal caching update strategy for the MEC server. The proposed scheme can capture task popularity, inter-task request correlation, users' communication conditions and computing capabilities. Simulations results show that the proposed method can rapidly converge to stable and precisely predict users' future task demands and outperform the other benchmark approaches in energy consumption.
In future work, we will optimize the bandwidth usage, time delay, and energy consumption under a practical MEC case with cloud-aided backhaul and asynchronous traffic.

\vspace{-0.8cm}
\appendix
\subsection{Proof of Lemma \ref{lem:NPhard}}\label{App:A1}
We prove that problem $\mathcal{P}$ is NP-hard via the restriction method \cite{garey1979computers}. Specifically, we show that the problem $\mathcal{P}$ can be restricted to a maximum cardinality bin packing problem. For clarity, we introduce the maximum cardinality bin packing problem \cite{loh2009solving}:
Given $K$ items with sizes $s_k$, $k \in \{1,2,\cdots, K\}$, and $M$ bins of identical capacity $Q$, the objective is to assign a maximum number of items to the fixed number of bins without violating the capacity constraint.

The NP-hardness of the maximum cardinality bin packing problem has been proved in \cite{loh2009solving}. To prove that Problem $\mathcal{P}$ is NP-hard, let us show that $\mathcal{P}$ contains a maximum cardinality bin packing problem as a special case. To this end, let us focus on one specific time slot $t$ by setting $ T = 1$, and assume that both the caching state of the MEC server $\bm{b}_t$ and the users' task request $\bm{\mu}_t$ are known. Thus, problem $\mathcal{P}$ is restricted as the following problem.
\begingroup
\setlength{\abovedisplayskip}{-4pt}
\setlength{\belowdisplayskip}{1pt}
\begin{align}
\mathcal{\widehat{P}}:~~\max_{\bm{\alpha}_t}~~& - \sum\nolimits_{k \in \mathcal{K}} E_{k,t}\label{prob:restrict1}\\
\text{s.~t.~~}& (\text{\ref{cons:one2}}), (\text{\ref{cons:one5}}).\notag
\end{align}
\endgroup
For problem $\mathcal{\widehat{P}}$, $\alpha_{k,t} = 0$ if and only if $E_{k,f}^{\text{L}} \le (1 - b_f^{(t)})E_{k,f,t}^{\text{O}} + b_f^{(t)}E_{k,f}^{{\text{C}}}$, otherwise user $k$ will select a channel to offload its task. Inspired by this, we further restrict problem $\mathcal{\widehat{P}}$ by setting $\alpha_{k,t} \in \mathcal{M}$ to just consider users execute their tasks through COMO. Additionally, we regard all users' COMO cost as -1 (i.e., $(1 - b_f^{(t)})E_{k,f,t}^{\text{O}} + b_f^{(t)} E_{k,f,t}^{\text{C}} = -1$) and each user request to execute a task $\mu_k^{(t)} \in \mathcal{F}$.  For ease of proof, we introduce a binary variable $\alpha_{k,m}^{(t)}$, where $\alpha_{k,m}^{(t)}= 1$ if and only if $\alpha_{k,t} = m$, otherwise is 0. Thus, we reformulate the restricted problem $\mathcal{\widehat{P}}$ as follows.
\begingroup
\setlength{\abovedisplayskip}{-4pt}
\setlength{\belowdisplayskip}{0pt}
\begin{align}
\widetilde{\mathcal{P}}:~~\max_{\bm{\alpha}_t}~~&\sum\nolimits_{k \in \mathcal{K}}  \sum\nolimits_{m \in \mathcal{M}} \alpha_{k,m}^{(t)}\label{prob:restrict2}\\
\text{s.~t.~~}& \sum\nolimits_{m \in \mathcal{M}} {\alpha_{k,m}^{(t)} \le 1} ,\tag{\theequation a}\\
&\sum\nolimits_{k \in \mathcal{K}} {\alpha_{k,m}^{(t)}{p_k}{h_k}}  \le Q,\tag{\theequation b}\\
&\alpha_{k,m}^{(t)} \in \left\{ {0,1} \right\},\tag{\theequation c}
\end{align}
\endgroup
where the capacity $Q$ is
\vspace{-0.3cm}
\begin{equation}\label{eq:Qtemp}
Q = \frac{p_k h_k}{{{2^{\frac{p_k \tau^2 (I_f + D_f - b_f^{(t)} D_f)}{B \zeta S_f^3}}} - 1}} - {\sigma ^2} + {p_k}{h_k}.
\end{equation}
Note that (\ref{eq:Qtemp}) follows from (\ref{ieq:off_not}).
For the restricted problem $\widetilde{\mathcal{P}}$, we regard the items and the bins in the maximum cardinality bin packing problem as the users and channels in problem $\mathcal{P}$, respectively. The size of item $k$ is $s_k = p_k h_k$. The objective of problem $\widetilde{\mathcal{P}}$ is to assign a maximum number of items to the fixed number of bins and satisfy the capacity constraint.
Thus, if problem $\widetilde{\mathcal{P}}$ can be effectively solved, the maximum cardinality bin packing problem can also be solved by a polynomial time algorithm. This manifests that the original problem $\mathcal{P}$ can be reduced to a maximum cardinality bin packing problem. Therefore, we can conclude that problem $\mathcal{P}$ is NP-hard.
\subsection{Proof of Remark \ref{rem:NE_exist}}\label{App:A2}
For user $k$, when the COMO decisions of other users except user $k$ (i.e., $\alpha_{-k,t}$) are given, we use $\alpha_{k,t}$ and $\alpha_{k,t}'$ to denote two different task offloading decisions of user $k$. Based on the definition of ordinal potential game in \cite{yamamoto2015comprehensive}, game $\bm{G}$ should satisfy
\begin{equation}\label{eq:conditonPG}
\setlength{\abovedisplayskip}{0pt}
\setlength{\belowdisplayskip}{0pt}
\mathop{\text{sgn}} [\phi (\alpha_{k,t},\alpha_{- k,t}) - \phi (\alpha_{k,t}',\alpha_{- k,t})]
= \mathop{\text{sgn}} [{f_t}({\alpha_{k,t}},{\alpha_{- k,t}}) - {f_t}(\alpha_{k,t}' ,{\alpha_{- k,t}})],
\end{equation}
where ${\mathop{\text{sgn}}} [\cdot]$ is a signum function. For ease of proof, we first derive the expression of $\phi (\alpha_{k,t},\alpha_{- k,t})$ as follows.
\begingroup
\setlength{\abovedisplayskip}{0pt}
\setlength{\belowdisplayskip}{0pt}
\begin{align}\label{eq:phi_rewrite}
&\phi(\alpha_{k,t},\alpha_{- k,t}) \notag =\frac{1}{2}\sum\limits_{k = 1}^K \sum\limits_{n \ne k} {p_k h_k p_n h_n \mathbbm{1}({\alpha_{n,t}} = {\alpha_{k,t}})} \mathbbm{1}(\alpha_{k,t} > 0) + \sum\limits_{k = 1}^K {p_k}{h_k}{V_k}\mathbbm{1}(\alpha_{k,t} = 0) \notag\\[-0.2cm]
&= \frac{1}{2}\sum\nolimits_{n \ne k} {{p_k h_k p_n h_n}\mathbbm{1}({\alpha_{n,t}} = {\alpha_{k,t}})}\mathbbm{1}({\alpha_{k,t}} > 0) + \frac{1}{2}\sum\nolimits_{l \ne k}^K {p_l h_l p_k h_k}\mathbbm{1}({\alpha_{k,t}} = {\alpha_{l,t}})\mathbbm{1}({\alpha_{l,t}} > 0) \notag\\[-0.2cm]
&+ \frac{1}{2}\sum\limits_{l \ne k}^K \sum\limits_{n \ne l,n \ne k} {{p_l h_l p_n h_n}\mathbbm{1}({\alpha_{n,t}} = {\alpha_{l,t}})} \mathbbm{1}({\alpha_{l,t}} > 0) + {p_k h_k V_k}\mathbbm{1}({\alpha_{k,t}} = 0) + \sum\limits_{l \ne k}^K {p_l h_l V_l}\mathbbm{1}({\alpha_{l,t}} = 0) \notag\\[-0.2cm]
&= {p_k h_k}\sum\limits_{n \ne k} {{p_n h_n}\mathbbm{1}({\alpha_{n,t}} = {\alpha_{k,t}})} \mathbbm{1}({\alpha_{k,t}} > 0) + \frac{1}{2}\sum\limits_{l \ne k}^K \sum\limits_{n \ne l,n \ne k} {{p_l h_l p_n h_n}\mathbbm{1}({\alpha_{n,t}} = {\alpha_{l,t}})} \mathbbm{1}({\alpha_{l,t}} > 0) \notag\\[-0.2cm]
&+ {p_k h_k V_k}\mathbbm{1}({\alpha_{k,t}} = 0) + \sum\nolimits_{l \ne k}^K {p_l h_l V_l}\mathbbm{1}({\alpha_{l,t}} = 0).
\end{align}
\endgroup
Below we discuss the relationship between $\phi(\alpha_{k,t},\alpha_{- k,t}) - \phi(\alpha_{k,t}' ,\alpha_{- k,t})$ and $f_t(\alpha_{k,t},\alpha_{- k,t}) - f_t(\alpha_{k,t}',\alpha_{- k,t})$ in three cases.
\begin{enumerate}[label={\arabic*)}]
\item ${\alpha_{k,t}} > 0,\alpha_{k,t}' > 0$. According to (\ref{eq:phi_rewrite}), we have
\begingroup
\setlength{\abovedisplayskip}{-4pt}
\setlength{\belowdisplayskip}{-1pt}
\begin{align}\label{eq:diff_phi}
\phi ({\alpha_{k,t}},{\alpha_{- k,t}}) - \phi (\alpha_{k,t}',{\alpha_{- k,t}}) &={p_k h_k}\sum\limits_{n \ne k} {{p_n h_n}\mathbbm{1}({\alpha_{n,t}} = {\alpha_{k,t}})} - {p_k h_k}\sum\limits_{n \ne k} {{p_n h_n}\mathbbm{1}({\alpha_{n,t}} = {\alpha_{k,t}}' )} \notag\\[-0.2cm]
&=p_k h_k\left(\Upsilon_{k,t} - \Upsilon_{k,t}' \right).
\end{align}
\endgroup
Based on (\ref{eq:energy_kt}), we have
\begin{equation}
\setlength{\abovedisplayskip}{1pt}
\setlength{\belowdisplayskip}{1pt}
f({\alpha_{k,t}},{\alpha_{- k,t}}) - f(\alpha _{k,t}',{\alpha_{- k,t}})= \sum\limits_{f \in \mathcal{F}} \mathbbm{1}(\mu_k^{(t)} = f) {p_k}(I_f + D_f - b_f^{(t)} D_f) ( \frac{1}{r_{k,t}} - \frac{1}{r_{k,t}'}).
\end{equation}
According to the definition of uplink rate and channel interference in (\ref{eq:link_rate}) and (\ref{ieq:off_not}), ${\mathop{\text{sgn}}} ( \frac{1}{r_{k,t}} - \frac{1}{r_{k,t}'} ) = \mathop{\text{sgn}}( \Upsilon_{k,t} - \Upsilon_{k,t}' )$ is established. Hence,  Eq. (\ref{eq:conditonPG}) is established in this case.

\item ${\alpha_{k,t}} > 0,\alpha_{k,t}' = 0$. Similarly, according to (\ref{eq:phi_rewrite}), we have
\begin{equation}\label{eq:diff_phi1}
\setlength{\abovedisplayskip}{1pt}
\setlength{\belowdisplayskip}{1pt}
\phi(\alpha_{k,t},\alpha_{- k,t}) \!-\! \phi (\alpha_{k,t}',\alpha_{- k,t}) \!=\! {p_k h_k} \bigg{(}\sum\limits_{n \ne k} {{p_n h_n}\mathbbm{1}({\alpha_{n,t}} \!=\! {\alpha_{k,t}})}  \!-\! {V_k}\bigg{)} \!=\!{p_k h_k}(\Upsilon_{k,t} \!-\! V_k ).
\end{equation}
Furthermore, according to (\ref{eq:energy_kt}), we have
\begingroup
\setlength{\abovedisplayskip}{0pt}
\setlength{\belowdisplayskip}{0pt}
\begin{equation}
f({\alpha_{k,t}},{\alpha_{ - k,t}}) - f(\alpha _{k,t}',{\alpha_{- k,t}}) =\sum\nolimits_{f \in \mathcal{F}} \mathbbm{1}(\mu_k^{(t)} = f)\Big{(} {p_k \frac{I_f + D_f - b_f^{(t)} D_f}{r_{k,t}} - \zeta \frac{S_f^3}{\tau^2}} \Big{)}.
\end{equation}
\endgroup
According to the analysis of (\ref{ieq:off_not}), we have ${\mathop{\text{sgn}}} \left( {\Upsilon_{k,t} - V_k} \right) = {\mathop{\text{sgn}}} ({p_k \frac{I_f + D_f - b_f^{(t)} D_f}{r_{k,t}} - \zeta \frac{S_f^3}{\tau^2}} )$. Thus, Eq. (\ref{eq:conditonPG}) is established in this case.

\item ${\alpha_{k,t}} = 0,\alpha_{k,t}' > 0$. This case is similar with case 2. Eq. (\ref{eq:conditonPG}) is also established in this case.
\end{enumerate}
Summarize the above results, Eq. (\ref{eq:conditonPG}) is established in any case. Consequently, game $\bm{G}$ is a ordinal potential game and can achieve a NE solution after finite number of iterations \cite{yamamoto2015comprehensive}.
\vspace{-0.6cm}
\subsection{Proof of Lemma \ref{lem:gameComplexity}}\label{App:A3}
For ease of presentation, we define $\Delta_{\max }= \mathop {\max}\nolimits_{k \in \mathcal{K}} \left\{p_k h_k\right\}$, $\Delta_{\min} = \mathop {\min}\nolimits_{k \in \mathcal{K}} \left\{ p_k h_k \right\}$, $V_{\max} = \mathop {\max }\nolimits_{k \in \mathcal{K}} \left\{ V_k\right\}$, $V_{\min} = \mathop {\min}\nolimits_{k \in \mathcal{K}} \left\{ V_k \right\}$. For the potential function, we have
\begin{align}
&\phi(\bm{\alpha}_t) \overset{(a)}= \frac{1}{2}\sum\nolimits_{k = 1}^K  \sum\nolimits_{n \ne k} {p_k h_k p_n h_n \mathbbm{1}(\alpha_{n,t} = \alpha_{k,t})} \mathbbm{1}(\alpha_{k,t} > 0) + \sum\nolimits_{k = 1}^K  {p_k}{h_k}{V_k}\mathbbm{1}({\alpha_{k,t}} = 0) \notag \\
&\le \frac{1}{2}\sum\nolimits_{k = 1}^K \sum\nolimits_{n \ne k} {\Delta_{\max }^2 \mathbbm{1}(\alpha_{n,t} = \alpha_{k,t})} \mathbbm{1}(\alpha_{k,t} > 0) + \sum\nolimits_{k = 1}^K  {\Delta_{\max}}{V_{\max}}\mathbbm{1}(\alpha_{k,t} = 0) \notag \\
&\le \frac{1}{2}{K^2}\Delta_{\max}^2 + K{\Delta_{\max }}{V_{\max}},
\end{align}
where (a) follows from (\ref{eq:potential_f}).

The COMO algorithm first initializes the COMO decisions of all users as 0, the initial value of $\phi(\bm{\alpha}_t)$ is $\phi(0) = \sum\nolimits_{k = 1}^K p_k h_k V_k \ge K{\Delta_{\min}}{V_{\min}}$.
Thus, the value range of $\phi(\bm{\alpha}_t)$ is less than $\frac{1}{2}{K^2}\Delta_{\max }^2 + K({\Delta_{\max}}{V_{\max}} - {\Delta_{\min}}{V_{\min}})$.
In each iteration, there is one user to update its decision to decrease the computing cost. Based on the definition of potential game, the decision update also decreases the value of potential function. It is assumed that user $k$ updates its offloading decision $\alpha_{k,t}$ to a better decision $\alpha_{k,t}'$ in one iteration, i.e., $\phi (\alpha_{k,t},\alpha_{- k,t}) - \phi (\alpha_{k,t}' ,\alpha_{- k,t}) > 0$. Below we analyze the decrement of $\phi(\bm{\alpha}_t)$ in each iteration in three cases.
\begin{enumerate}[label={\arabic*)}]
\item ${\alpha_{k,t}} > 0$ and $\alpha_{k,t}' > 0$.
\begingroup
\setlength{\abovedisplayskip}{0pt}
\setlength{\belowdisplayskip}{1pt}
\begin{align}
\phi(\alpha_{k,t},\alpha_{- k,t}) \!-\! \phi (\alpha_{k,t}' ,\alpha_{- k,t}) \overset{(a)}=  {p_k}{h_k}\sum\limits_{n \ne k} {p_n h_n} \Big{(}\mathbbm{1}(\alpha_{k,t} \!=\! \alpha_{n,t}) \!-\! \mathbbm{1}(\alpha_{k,t}' \!=\! \alpha_{n,t}) \Big{)} >0.
\end{align}
\endgroup
where (a) follows from (\ref{eq:diff_phi}). Since the value of indicator function $\mathbbm{1}(\cdot)$ is integer, we have
\begin{align}
\sum\nolimits_{n \ne k} {p_n h_n} \left( \mathbbm{1}(\alpha_{k,t} = \alpha_{n,t}) - \mathbbm{1}(\alpha_{k,t}'  = \alpha_{n,t}) \right) \ge {\Delta_{\min}}.
\end{align}
Consequently, $\phi ({\alpha_{k,t}},{\alpha_{- k,t}}) - \phi (\alpha_{k,t}',{\alpha_{- k,t}}) \ge \Delta_{\min }^2$.

\item $\alpha_{k,t} > 0,\alpha_{k,t}' = 0$.
\begin{align}
\phi ({\alpha_{k,t}},{\alpha_{- k,t}}) - \phi (\alpha_{k,t}',{\alpha_{-k,t}}) \overset{(a)}= {p_k h_k}\Big{(} {\sum\nolimits_{n \ne k} {{p_n h_n}\mathbbm{1}(\alpha_{n,t} = \alpha_{k,t})} - V_k} \Big{)} > 0.
\end{align}
where (a) follows from (\ref{eq:diff_phi1}).
Thus, there is a positive number $\varepsilon = \sum\nolimits_{n \ne k} p_n h_n \mathbbm{1}(\alpha_{n,t} = \alpha_{k,t}) - V_k$, subject to $\phi ({\alpha_{k,t}},{\alpha_{- k,t}}) - \phi ({\alpha_{k,t}}',{\alpha_{- k,t}}) = \varepsilon {p_k}{h_k} \ge \varepsilon {\Delta_{\min }}$

\item $\alpha_{k,t} = 0,\alpha_{k,t}' > 0$. Similar to case 2, there is a positive integer $\varepsilon$ such that $\phi(\alpha_{k,t},\alpha_{-k,t}) - \phi (\alpha_{k,t}',\alpha_{-k,t}) \ge \varepsilon {\Delta_{\min}}$.

\end{enumerate}

Summarizing the above three cases, we have $\phi ({\alpha_{k,t}},{\alpha_{- k,t}}) - \phi (\alpha_{k,t}',\alpha_{-k,t}) \ge \varepsilon {\Delta_{\min }}$, where $\varepsilon $ is a positive number. That is to say, in each iteration, the potential function will decrease at least $\varepsilon {\Delta_{\min }}$. Accordingly, the algorithm will terminate within $\frac{{\frac{1}{2}{K^2}\Delta_{\max}^2 + K(\Delta_{\max} V_{\max} - \Delta_{\min} V_{\min})}}{{\varepsilon \Delta_{\min}}}$ iterations and obtain a NE solution for COMO problem.

\bibliographystyle{IEEEtran}
\bibliography{IEEEabrv,cited}

\begin{thebibliography}{10}
\providecommand{\url}[1]{#1}
\csname url@rmstyle\endcsname
\providecommand{\newblock}{\relax}
\providecommand{\bibinfo}[2]{#2}
\providecommand\BIBentrySTDinterwordspacing{\spaceskip=0pt\relax}
\providecommand\BIBentryALTinterwordstretchfactor{4}
\providecommand\BIBentryALTinterwordspacing{\spaceskip=\fontdimen2\font plus
\BIBentryALTinterwordstretchfactor\fontdimen3\font minus
  \fontdimen4\font\relax}
\providecommand\BIBforeignlanguage[2]{{%
\expandafter\ifx\csname l@#1\endcsname\relax
\typeout{** WARNING: IEEEtran.bst: No hyphenation pattern has been}%
\typeout{** loaded for the language `#1'. Using the pattern for}%
\typeout{** the default language instead.}%
\else
\language=\csname l@#1\endcsname
\fi
#2}}

\bibitem{9363323}
Y.~Siriwardhana, P.~Porambage, M.~Liyanage, and M.~Ylianttila, ``A survey on
  mobile augmented reality with 5{G} mobile edge computing: Architectures,
  applications, and technical aspects,'' \emph{IEEE Commun. Surveys Tuts.},
  vol.~23, no.~2, pp. 1160--1192, 2021.

\bibitem{mao2017survey}
Y.~Mao, C.~You, J.~Zhang, K.~Huang, and K.~B. Letaief, ``A survey on mobile
  edge computing: The communication perspective,'' \emph{IEEE Commun. Surveys
  Tuts.}, vol.~19, no.~4, pp. 2322--2358, 2017.

\bibitem{sabella2016mobile}
D.~Sabella, A.~Vaillant, P.~Kuure, U.~Rauschenbach, and F.~Giust, ``Mobile-edge
  computing architecture: The role of mec in the internet of things,''
  \emph{IEEE Consum. Electron. Mag.}, vol.~5, no.~4, pp. 84--91, 2016.

\bibitem{mach2017mobile}
P.~Mach and Z.~Becvar, ``Mobile edge computing: A survey on architecture and
  computation offloading,'' \emph{IEEE Commun. Surveys Tuts.}, vol.~19, no.~3,
  pp. 1628--1656, 2017.

\bibitem{sun2020online}
Z.~Sun and M.~R. Nakhai, ``An online learning algorithm for distributed task
  offloading in multi-access edge computing,'' \emph{IEEE Trans. Signal
  Processing}, vol.~68, pp. 3090--3102, 2020.

\bibitem{alameddine2019dynamic}
H.~A. Alameddine, S.~Sharafeddine, S.~Sebbah, S.~Ayoubi, and C.~Assi, ``Dynamic
  task offloading and scheduling for low-latency {IoT} services in multi-access
  edge computing,'' \emph{IEEE J. Selected Areas Commun.}, vol.~37, no.~3, pp.
  668--682, 2019.

\bibitem{yu2020joint}
Z.~Yu, Y.~Gong, S.~Gong, and Y.~Guo, ``Joint task offloading and resource
  allocation in {UAV}-enabled mobile edge computing,'' \emph{IEEE Internet
  Things J.}, vol.~7, no.~4, pp. 3147--3159, 2020.

\bibitem{zhang2020dynamic}
Q.~Zhang, L.~Gui, F.~Hou, J.~Chen, S.~Zhu, and F.~Tian, ``Dynamic task
  offloading and resource allocation for mobile-edge computing in dense cloud
  {RAN},'' \emph{IEEE Internet Things J.}, vol.~7, no.~4, pp. 3282--3299, 2020.

\bibitem{zhao2021energy}
M.~Zhao, J.-J. Yu, W.-T. Li, D.~Liu, S.~Yao, W.~Feng, C.~She, and T.~Q. Quek,
  ``Energy-aware task offloading and resource allocation for time-sensitive
  services in mobile edge computing systems,'' \emph{IEEE Trans. Veh.
  Technol.}, vol.~70, no.~10, pp. 10\,925--10\,940, 2021.

\bibitem{8467992}
T.~X. Tran and D.~Pompili, ``Adaptive bitrate video caching and processing in
  mobile-edge computing networks,'' \emph{IEEE Trans. Mobile Computing},
  vol.~18, no.~9, pp. 1965--1978, 2019.

\bibitem{8314696}
M.~Chen and Y.~Hao, ``Task offloading for mobile edge computing in software
  defined ultra-dense network,'' \emph{IEEE J. Sel. Areas Commun.}, vol.~36,
  no.~3, pp. 587--597, 2018.

\bibitem{wen2020joint}
W.~Wen, Y.~Cui, T.~Q. Quek, F.-C. Zheng, and S.~Jin, ``Joint optimal software
  caching, computation offloading and communications resource allocation for
  mobile edge computing,'' \emph{IEEE Trans. Veh. Technol.}, vol.~69, no.~7,
  pp. 7879--7894, 2020.

\bibitem{yan2021pricing}
J.~Yan, S.~Bi, L.~Duan, and Y.-J.~A. Zhang, ``Pricing-driven service caching
  and task offloading in mobile edge computing,'' \emph{IEEE Trans. Wireless
  Commun.}, vol.~20, no.~7, pp. 4495--4512, 2021.

\bibitem{9509427}
Z.~Chen, Z.~Zhou, and C.~Chen, ``Code caching-assisted computation offloading
  and resource allocation for multi-user mobile edge computing,'' \emph{IEEE
  Trans. Netw. Service Manag.}, vol.~18, no.~4, pp. 4517--4530, 2021.

\bibitem{9076825}
S.~Bi, L.~Huang, and Y.-J.~A. Zhang, ``Joint optimization of service caching
  placement and computation offloading in mobile edge computing systems,''
  \emph{IEEE Trans. Wireless Commun.}, vol.~19, no.~7, pp. 4947--4963, 2020.

\bibitem{8401954}
W.~Yi, Y.~Liu, and A.~Nallanathan, ``Cache-enabled hetnets with millimeter wave
  small cells,'' \emph{IEEE Trans. Commun.}, vol.~66, no.~11, pp. 5497--5511,
  Nov. 2018.

\bibitem{xing2019energy}
H.~Xing, J.~Cui, Y.~Deng, and A.~Nallanathan, ``Energy-efficient proactive
  caching for fog computing with correlated task arrivals,'' in \emph{in Proc.
  SPAWC}.\hskip 1em plus 0.5em minus 0.4em\relax IEEE, 2019, pp. 1--5.

\bibitem{yang2019joint}
X.~Yang, Z.~Fei, J.~Zheng, N.~Zhang, and A.~Anpalagan, ``Joint multi-user
  computation offloading and data caching for hybrid mobile cloud/edge
  computing,'' \emph{IEEE Trans. Veh. Technol.}, vol.~68, no.~11, pp.
  11\,018--11\,030, 2019.

\bibitem{chen2020dynamic}
Z.~Chen and Z.~Zhou, ``Dynamic task caching and computation offloading for
  mobile edge computing,'' in \emph{Proc. IEEE GLOBECOM}.\hskip 1em plus 0.5em
  minus 0.4em\relax IEEE, 2020, pp. 1--6.

\bibitem{9013927}
Z.~Chen, Z.~Chen, and Y.~Jia, ``Integrated task caching, computation offloading
  and resource allocation for mobile edge computing,'' in \emph{Proc. IEEE
  GLOBECOM}, 2019, pp. 1--6.

\bibitem{bi2020joint}
S.~Bi, L.~Huang, and Y.-J.~A. Zhang, ``Joint optimization of service caching
  placement and computation offloading in mobile edge computing systems,''
  \emph{IEEE Trans. Wireless Commun.}, vol.~19, no.~7, pp. 4947--4963, 2020.

\bibitem{zhang2018joint}
J.~Zhang, X.~Hu, Z.~Ning, E.~C.-H. Ngai, L.~Zhou, J.~Wei, J.~Cheng, B.~Hu, and
  V.~C. Leung, ``Joint resource allocation for latency-sensitive services over
  mobile edge computing networks with caching,'' \emph{IEEE Internet Things
  J.}, vol.~6, no.~3, pp. 4283--4294, 2018.

\bibitem{8778670}
P.~Wu, J.~Li, L.~Shi, M.~Ding, K.~Cai, and F.~Yang, ``Dynamic content update
  for wireless edge caching via deep reinforcement learning,'' \emph{IEEE
  Commun. Lett.}, vol.~23, no.~10, pp. 1773--1777, 2019.

\bibitem{9110932}
Y.~Qian, R.~Wang, J.~Wu, B.~Tan, and H.~Ren, ``Reinforcement learning-based
  optimal computing and caching in mobile edge network,'' \emph{IEEE J. Sel.
  Areas Commun.}, vol.~38, no.~10, pp. 2343--2355, 2020.

\bibitem{8491367}
J.~Zhang, X.~Hu, Z.~Ning, E.~C.-H. Ngai, L.~Zhou, J.~Wei, J.~Cheng, B.~Hu, and
  V.~C.~M. Leung, ``Joint resource allocation for latency-sensitive services
  over mobile edge computing networks with caching,'' \emph{IEEE Internet
  Things J.}, vol.~6, no.~3, pp. 4283--4294, 2019.

\bibitem{9419788}
R.~Zheng, H.~Wang, M.~De~Mari, M.~Cui, X.~Chu, and T.~Q.~S. Quek, ``Dynamic
  computation offloading in ultra-dense networks based on mean field games,''
  \emph{IEEE Trans. Wireless Commun.}, vol.~20, no.~10, pp. 6551--6565, 2021.

\bibitem{8418400}
Y.~Sun, Y.~Cui, and H.~Liu, ``Joint pushing and caching for bandwidth
  utilization maximization in wireless networks,'' \emph{IEEE Trans. Commun.},
  vol.~67, no.~1, pp. 391--404, 2019.

\bibitem{rappaport1996wireless}
T.~S. Rappaport \emph{et~al.}, \emph{Wireless communications: principles and
  practice}.\hskip 1em plus 0.5em minus 0.4em\relax prentice hall PTR New
  Jersey, 1996, vol.~2.

\bibitem{7307234}
X.~Chen, L.~Jiao, W.~Li, and X.~Fu, ``Efficient multi-user computation
  offloading for mobile-edge cloud computing,'' \emph{IEEE/ACM Trans. Netw.},
  vol.~24, no.~5, pp. 2795--2808, 2016.

\bibitem{1194818}
M.~Xiao, N.~Shroff, and E.~Chong, ``A utility-based power-control scheme in
  wireless cellular systems,'' \emph{IEEE/ACM Trans. Netw.}, vol.~11, no.~2,
  pp. 210--221, 2003.

\bibitem{chiang2008power}
M.~Chiang, P.~Hande, T.~Lan, C.~W. Tan, \emph{et~al.}, ``Power control in
  wireless cellular networks,'' \emph{Found. Trends Netw.}, vol.~2, no.~4, pp.
  381--533, 2008.

\bibitem{miettinen2010energy}
A.~P. Miettinen and J.~K. Nurminen, ``Energy efficiency of mobile clients in
  cloud computing.'' \emph{HotCloud}, vol.~10, pp. 1--7, 2010.

\bibitem{9340334}
A.~Bozorgchenani, D.~Tarchi, and W.~Cerroni, ``On-demand service deployment
  strategies for fog-as-a-service scenarios,'' \emph{IEEE Commun. Letters},
  vol.~25, no.~5, pp. 1500--1504, 2021.

\bibitem{van2016deep}
H.~Van~Hasselt, A.~Guez, and D.~Silver, ``Deep reinforcement learning with
  double {Q}-learning,'' in \emph{in Proc AAAI}, vol.~30, no.~1, 2016.

\bibitem{sutton2018reinforcement}
R.~S. Sutton and A.~G. Barto, \emph{Reinforcement learning: An
  introduction}.\hskip 1em plus 0.5em minus 0.4em\relax MIT press, 2018.

\bibitem{agarwal2020optimistic}
R.~Agarwal, D.~Schuurmans, and M.~Norouzi, ``An optimistic perspective on
  offline reinforcement learning,'' in \emph{in Proc. ICML}, 2020, pp.
  104--114.

\bibitem{goodfellow2016deep}
I.~Goodfellow, Y.~Bengio, and A.~Courville, \emph{Deep learning}.\hskip 1em
  plus 0.5em minus 0.4em\relax MIT press, 2016.

\bibitem{garey1979computers}
M.~R. Garey and D.~S. Johnson, \emph{Computers and intractability}.\hskip 1em
  plus 0.5em minus 0.4em\relax freeman San Francisco, 1979, vol. 174.

\bibitem{8362880}
G.~Hasslinger, J.~Heikkinen, K.~Ntougias, F.~Hasslinger, and O.~Hohlfeld,
  ``Optimum caching versus {LRU} and {LFU}: Comparison and combined limited
  look-ahead strategies,'' in \emph{in Proc. WiOpt}, 2018, pp. 1--6.

\bibitem{loh2009solving}
K.-H. Loh, B.~Golden, and E.~Wasil, ``Solving the maximum cardinality bin
  packing problem with a weight annealing-based algorithm,'' in
  \emph{Operations Research and Cyber-Infrastructure}.\hskip 1em plus 0.5em
  minus 0.4em\relax Springer, 2009, pp. 147--164.

\bibitem{yamamoto2015comprehensive}
K.~Yamamoto, ``A comprehensive survey of potential game approaches to wireless
  networks,'' \emph{IEICE Trans. Commun.}, vol.~98, no.~9, pp. 1804--1823,
  2015.

\end{thebibliography}
\end{document}